\documentclass[12pt]{article}
\RequirePackage[colorlinks,citecolor=blue,linkcolor=blue,urlcolor=blue]{hyperref}

\usepackage[utf8]{inputenc}
\usepackage{amsfonts}
\usepackage{amsmath} 
\usepackage{amsthm}
\usepackage{amssymb}
\usepackage{amsbsy}
\usepackage{apacite} 
\usepackage{apptools}
    \AtAppendix{\counterwithin{lemma}{section}}
    \AtAppendix{\counterwithin{corollary}{section}}
    \AtAppendix{\counterwithin{assumption}{section}}
    \AtAppendix{\counterwithin{theorem}{section}}
    \AtAppendix{\counterwithin{proposition}{section}}
\usepackage{authblk}
\usepackage{booktabs}
\usepackage[a4paper, bindingoffset = 0.2in, left = 1in, right = 1in, top = 1in, bottom = 1in, footskip = 0.25in]{geometry}
\usepackage{calrsfs}
\usepackage{caption}
\usepackage[onehalfspacing]{setspace}
\usepackage{pgfplots}
\usepackage{tikz}
     \usetikzlibrary{positioning}
    \usetikzlibrary{calc}
    \usetikzlibrary{arrows}
\usepackage{titling}
    
\DeclareMathAlphabet{\pazocal}{OMS}{zplm}{m}{n}
\usepackage{tgpagella}
\usepackage{yfonts}

\usepackage{graphicx}
\usepackage{caption} % for notes under figures
\usepackage{subcaption}

\newtheoremstyle{assumption}{7.5pt}{7.5pt}{}{}{\itshape}{}{1em}{}
\newtheoremstyle{main}{10pt}{10pt}{\addtolength{\leftskip}{1.5em}}{}{\bfseries}{}{1em}{}

\theoremstyle{plain}
\newtheorem{theorem}{Theorem}

\newtheorem{corollary}{Corollary}
\newtheorem{lemma}{Lemma}

\newtheorem{proposition}{Proposition}
\newtheorem{assumption}{Assumption}

\newcommand{\covariance}{\mathbb{C}}

\newcommand{\expectation}{\mathbb{E}}
\newcommand{\variance}{\mathbb{V}}
\newcommand{\plim}{\mathrm{plim}~}
\newcommand{\diag}{\mathrm{{diag}}}
\newcommand{\trace}{\mathrm{{trace}}}
\newcommand{\rank}{\mathrm{{rank}}}

\newcommand{\vect}{\mathrm{{vec}}}

\newcommand{\boldy}{\mathbf{{y}}}

\newcommand{\boldG}{\mathbf{{G}}}

\newcommand{\bG}{\mathbf{{G}}}
\newcommand{\bI}{\mathbf{{I}}}
\newcommand{\bJ}{\mathbf{{J}}}
\newcommand{\bA}{\mathbf{{A}}}

\newcommand{\bC}{\mathbf{{C}}}

\newcommand{\bM}{\mathbf{{M}}}

\newcommand{\bm}{\mathbf{{m}}}

\newcommand{\bx}{\mathbf{x}}

\newcommand{\bu}{\mathbf{{u}}}

\newcommand{\bnx}{\widetilde{\mathbf{x}}}

\newcommand{\nx}{\widetilde{{x}}}

\newcommand{\gammaOLS}{\widehat{\gamma}^{OLS}}
\newcommand{\deltaOLS}{\widehat{\delta}^{OLS}}

\newcommand{\alphaIV}{\widehat{\alpha}^{IV}}
\newcommand{\gammaIV}{\widehat{\gamma}^{IV}}
\newcommand{\deltaIV}{\widehat{\delta}^{IV}}
\newcommand{\betaIV}{\widehat{\beta}^{IV}}

\newcommand{\bW}{\mathbf{W}}

\newcommand{\be}{\mathbf{e}}

\newcommand{\boldS}{\mathbf{{S}}}

\newcommand{\boldSigma}{\boldsymbol{\Sigma}}

\title{Measurement Error and Peer Effects \\ in Networks \protect\footnotetext{\llap{\textsuperscript{*}}Yann acknowledges financial support from the French government under the “France 2030” investment plan managed by the French National Research Agency (ANR-17-EURE-0020) and from the Excellence Initiative of Aix-Marseille University - A*MIDEX. Sebastiaan benefited from doctoral and postdoctoral fellowships from the Research Foundation Flanders (grants 11F8919N and 12C8623N). We thank the editor, \'Aureo de Paula, an anonymous associate editor, and two anonymous reviewers for very helpful comments and suggestions. We also thank Vincent Boucher, Nathan Canen, Sam Cosaert, Geert Dhaene, Marcel Fafchamps, Jan Feld, Bernard Fortin, Alan Griffith, as well as conference and seminar participants in Antwerp, Barcelona (ESEM), Manchester (EWMES), Orl\'eans (IPDC), Oslo (IAAE), Québec, and Rome. Part of this research was conducted during visits to Aix-Marseille School of Economics, whose hospitality is gratefully acknowledged.}}

\author{{\sc Yann Bramoull\'e} \\ \small %Aix-Marseille School of Economics \\
\small Aix-Marseille University \& CNRS \\ \small \url{yann.bramoulle@univ-amu.fr}
  \and {\sc Sebastiaan Maes} \\ \small %Department of Economics \\ \small 
  University of Antwerp \& CESifo \\ \small \url{sebastiaan.maes@uantwerpen.be}
}
\vspace{100pt}
\date{\today }

\begin{document}

\begin{titlepage}
	\maketitle
	\thispagestyle{empty}
	\begin{abstract}
         In many practical applications, only noisy proxies for the true regressors are available, which is commonly believed to induce an attenuation bias. In the linear-in-means model, however, estimated peer effects might be inflated, potentially leading to false positives. This paper shows that the asymptotic bias depends on the interplay between individual characteristics and network links and demonstrates how the network structure can facilitate identification without the need for additional external information. Based on these identification results, we present consistent GMM and 2SLS estimators that are easily implementable. Our results are illustrated by means of a Monte Carlo simulation. 
	\end{abstract}

	\vspace{0.1 in}

	
	\vspace{0.1 in}
	
\end{titlepage}

%%%%%%%%%%%%%%%%%%%%%%%%%%%%%%%%%%%%%%%%%%%%%%%%%%%%%%%%%%%%%%%%%%%%%%%
%%%%%%%%%%%%%%%%%%%%%%%%%%%%%%%%%%%%%%%%%%%%%%%%%%%%%%%%%%%%%%%%%%%%%%%
%%%%%%%%%%%%%%%%%%%%%%%%%%%%%%%%%%%%%%%%%%%%%%%%%%%%%%%%%%%%%%%%%%%%%%%
\section{Introduction}
How reliable are peer effect estimates in published studies? Not very, according to Josh Angrist's critical assessment of the literature: ``although correlation among peers is a reliable descriptive fact, the scope for incorrect or misleading attributions of causality in peer analysis is extraordinarily wide" \cite[p.~98]{angrist_perils_2014}. Obtaining credible causal estimates of peer effects presents several challenges to researchers. To address peer endogeneity, they must find, or generate, exogenous variations in peer groups and develop models of peer selection.\footnote{For studies of peer effects with random peers see, e.g., \citeA{sacredote01, carrellsacredotewest13, Cornolaferraraburns22}. For econometric approaches combining models of peer effects in networks with models of network formation see, e.g., \citeA{goldsmithimbens13, hsiehlee, griffith22}.} To address mismeasurement in peers, they must collect detailed data on social networks and develop models of peer effects with unknown or mismeasured peers.\footnote{The literature on peer effects in networks has grown fast in the past fifteen years, see \citeA{bramoulledjebbarifortin20} for a review. See, e.g., \citeA{ griffith21, BoucherHoundetoungan23, lewbelqutang21, depaularasulsouza19} for models of peer effects with imperfectly known peers.} Building on \citeA{manski93}, the methodological literature on peer effects has grown alongside the applied literature. Yet despite methodological advances, guaranteeing the reliability of peer effect estimates remains challenging. Addressing multiple endogeneity issues without resolving them all still leaves researchers some distance away from a causal interpretation.

This paper focuses on a critical yet understudied issue affecting the reliability of peer effect estimates: measurement error in individual characteristics. Such errors are a first-order empirical issue in survey-based \cite{boundbrownmathiowetz} and experimental work \cite{gillensnowbergyariv19}. Errors-in-variables are especially salient when individual-level covariates are measured with noise due to survey misreporting, recall error, or imperfect measurement. In particular, self-reported income and assets, as well as noisy proxies for ability such as test scores, are commonly subject to classical measurement error.

Surprisingly, however, errors-in-variables remain a blind spot in the existing literature. In the applied literature on peer effects, problems raised by measurement error on covariates are almost never discussed or addressed. The methodological literature on this problem is scarce, and will be comprehensively reviewed below. In short, this issue was identified by \citeA{steven_policy_2001}, highlighted by \citeA{angrist_perils_2014}, and studied by \citeA{ammermueller_peer_2009, paula_2017, feld_understanding_2017}. These five papers, however, only consider group interactions: agents are partitioned in groups, such as classes within schools, and are affected by everyone in their group and by no one outside it. In contrast, many recent studies of peer effects consider richer network interactions. Our paper fills this gap by analyzing measurement error in network-based peer effects.

We provide the first analysis of measurement error when peer effects operate on a network. We consider the benchmark linear-in-means model with contextual and endogenous peer effects. An individual's outcome is affected by their individual characteristic, by the average characteristic of their peers, and by the average outcome of their peers. A significant challenge arises when  researchers only observe a noisy proxy for the characteristic. Measurement error then occurs in both individual and average peer characteristic, and errors in the two variables are related through the structure of the model. The presence of related measurement errors on two regressors makes the problem non-standard. We adopt a many-networks asymptotic framework and consider econometric specifications with or without network fixed effects. Our investigation revolves around two central questions. First, under what conditions does measurement error cause asymptotic biases in peer effect estimates? Second, what strategies can researchers employ to mitigate these biases? Our identification and estimation results highlight that true peer effects can generally be recovered without additional data. 

Consistent with existing research, our findings indicate that classical measurement error typically induces asymptotic bias in peer-effect estimates. Defying conventional wisdom, these estimates are not necessarily attenuated toward zero.\footnote{This conventional wisdom is based on the analysis of OLS estimates in the presence of one mismeasured variable. It is well known that little can be said, in general, about the nature of the asymptotic bias in OLS or 2SLS estimates if there is measurement error in multiple variables (see, e.g., \citeNP{levi73}; \citeNP[p.~86]{greene03}).} We show that naive 2SLS estimates of peer effects, using mismeasured characteristics of peers at distance two as an instrument for the average outcome of peers, converge to a linear combination of the model's parameters, which we refer to as \emph{smearing}.\footnote{\citeA{GarberKlepper80} distinguish three effects of classical measurement error on least-squares estimates: an ``own'' effect, whereby the coefficient on a mismeasured regressor is attenuated by its own error; a ``smearing'' effect, whereby the coefficient on a mismeasured regressor is biased by measurement error in \emph{another} mismeasured regressor; and a ``picking up'' effect, whereby the coefficients on correctly measured regressors absorb part of the mismeasured regressors' effects. For simplicity, we refer to all cross-coefficient biases as smearing.} The magnitude and direction of the resulting asymptotic biases critically depend on the interplay between individual characteristics and network links. With measurement error, the estimated endogenous peer effect absorbs parts of the individual effect and of the contextual peer effect, because the instrument spuriously picks up variation in the true individual and average peer characteristic. Similarly, the estimated contextual peer effect absorbs part of the individual effect, because the average mismeasured peer characteristic partially proxies for the true individual characteristic. As a consequence, peer effect estimates may be significantly different from zero even in the absence of true peer effects. These cross-loadings vanish when characteristics are independent of links and i.i.d. across individuals. In that case, the estimate of the endogenous peer effect is asymptotically unbiased, while estimates of individual and contextual peer effects display an asymptotic attenuation bias towards zero.

To build intuition we also study the bias in two simplified settings: a model with contextual peer effects only and a model with endogenous peer effects only. In both cases, smearing only arises when an individual's characteristic is correlated with their average peers' characteristics. This correlation naturally emerges under \textit{homophily}, which is widely documented in social networks \cite{McPherson}. For example, it is common for college students to form friendships with peers who have comparable academic abilities, similar socioeconomic backgrounds, and similar levels of parental education. Measures of these characteristics in survey data, such as \textit{Add Health}, are notoriously noisy, suggesting a widespread risk of biased peer effect estimates.\footnote{Although homophily disappears under complete randomization, it reappears in common quasi-experimental setups, such as peer randomization within stratified groups. Under complete randomization, inflated estimates might still occur when individuals are drawn from a finite pool, giving rise to the exclusion bias (e.g., see \citeNP{caeyersfafchamps23}).}

In the first part of our analysis, we thus clarify how relationships between characteristics and links give rise to an asymptotic bias in peer effect estimates. In the second part, we demonstrate how these relationships can be leveraged to solve the problem. We show that the econometric model can be identified through conditional mean and covariance restrictions, and without relying on external information.\footnote{External information, when available, can of course help better identify the parameters of interest. A standard method to address measurement error is to exploit multiple independent measurements of the noisy variable \cite{reiersol41, schennach07}. This can easily be combined with our internal identification strategies.} Identification depends on features of the measurement error. Under a weak assumption that measurement error has mean zero, we characterize when the model is identified based on conditional mean restrictions (Theorem~\ref{theorem:identification}). We derive a necessary and sufficient rank condition for identification, which combines the interaction matrix and the characteristic's first moments. This strategy exploits potential associations between network positions and the individual characteristic: for instance, when agents with more peers tend to have a higher value for the characteristic. Identification from first moments is generic if the network structure is rich enough (Proposition~\ref{proposition:genericity}) and holds even when the measurement error is correlated and heteroscedastic. 

Imposing more structure on the measurement error opens up more possibilities for identification. In a second step, we consider uncorrelated and homoscedastic measurement errors. We show when the model is identified based on conditional covariance restrictions (Theorem~\ref{thm:identification2}). This identification strategy takes advantage of variations in correlations across observed characteristics between two network positions. Overall, our results show that except in special cases, it is possible to eliminate the asymptotic bias resulting from measurement error.

To do so in practice, we propose generalized method of moments (GMM) and two-stage least squares (2SLS) estimators that are easy to implement. The GMM approach uses more information and is generally more efficient, albeit introducing non-linearities in the estimation procedure. Moments are directly built from the means and covariance restrictions highlighted in our identification results. In contrast, 2SLS estimation is linear but typically less efficient. We propose a variety of valid instruments, including network lags and network features. Monte Carlo simulations show that even modest measurement error can induce substantial bias in naive 2SLS estimates of peer effects. They also confirm that our proposed estimators are feasible and perform well in finite samples.

\paragraph{Related literature.} 
Following the work of \citeA{manski93}, the applied literature on peer effects has grown extensively.\footnote{For a theoretical and econometric discussion on linear social interaction models see, e.g., \citeA{blumebrockdurlaufjayaraman15}.} Peer effects have been studied in a wide variety of settings, ranging from the classroom  \cite{lavy}, through labor supply \cite{nicoletti}, to consumption decisions \cite{degiorgi19}. When it is suspected that characteristics may suffer from measurement error, our results and the tools proposed here can aid empirical researchers to address this problem and to obtain consistent peer effect estimates. Our analysis also clarifies the impact of randomization, a common strategy to mitigate peer endogeneity, on this bias.

Our analysis advances the sparse literature on measurement error and peer effects. A first strand of research concerns models with contextual peer effects. \citeA{steven_policy_2001} was the first to show formally that errors-in-variables can give rise to an expansion bias in peer effect estimates. He discusses the type of policy interventions that can help address the problem. \citeA{angrist_perils_2014} highlights the role played by measurement error in generating inflated peer effect estimates. He illustrates in Table 3 p.103 how adding noise to individual schooling leads to a large increase in the estimate of average state schooling, in a regression on log wage. Exploiting variation across classes within schools, \citeA{ammermueller_peer_2009} demonstrate that the inclusion of school fixed effects considerably reduces the magnitude of class peer effect estimates. This discrepancy is attributed to the interplay between errors-in-variables and homophily: when sorting into classes is random but sorting into schools is not, the inclusion of school fixed effects removes the school-level homophily that gives rise to the expansion bias. \citeA{feld_understanding_2017} show that when assignment to classes is completely random, errors-in-variables only lead to an attenuation bias. For the same setting, \citeA{feldzolitz22} propose a simple bias correction procedure based on multiple noisy measurements of the true characteristic.

Another strand of research concerns models with endogenous peer effects and noisy outcomes. \citeA[p.~310]{paula_2017} shows that the covariance between peers' mismeasured outcomes identifies the endogenous peer effect in a linear-in-means model. The result operates under the assumption of homoscedastic and uncorrelated disturbances in the outcome equation. In the context of a game with misclassified binary actions, \citeA{linhu23} develop a consistent estimator based on repeated measurements.

Taking a broader perspective, we also contribute to the literature on measurement error in dependent data. In a seminal contribution, \citeA{grilicheshausman86} show how the errors-in-variables problem can be overcome in the standard linear panel data model without resorting to outside information. Our use of network-lagged characteristics as instruments resembles their use of time-lagged variables. More recently, \citeA{evdokimovweleneev20} study errors-in-variables in general nonlinear semiparametric panel or network data models with fixed effects. In this more general setting, they show how the lagged values of covariates can still serve as instruments to overcome the bias. However, they assume that the variances of the measurement errors shrink with sample size, which is rather restrictive in our context.

Our work is also related to the literature on errors-in-variables in linear models through its use of higher moments. In early contributions, \citeA{koopmans37} and \citeA{reiersol50} recognized that this approach fails if the observables are jointly normal distributed. \citeA{cragg97}, \citeA{dagenais97}, and \citeA{ericksonwhited2002} therefore impose rank conditions on third and higher moments to ensure identification. \citeA{klepperleaner84} show that the first and second moments can be used to bound the coefficients. More recently, \citeA{benmoshe21} provides necessary and sufficient conditions for identification when there is measurement error in all variables. Alternatively, if some variables are known to be perfectly measured, the latter can be used to construct instruments. \citeA{lewbel97, lewbel12} and \citeA{benmoshehaultfoeuillelewbel17} construct valid instruments from perfectly measured variables without using additional outside information. Our approach differs from this literature in that we do not impose functional form assumptions on the distribution of measurement error, nor full independence of measurement error from the other variables in the model. Moreover, our results do not require the presence of perfectly measured covariates. We also allow for conditional heteroscedasticity in outcomes, which might be important in empirical applications.

\paragraph{Outline of the paper.}
The remainder of this paper is organized as follows. Section~\ref{sec:setup} introduces the linear-in-means model and details the associated naive 2SLS estimator. In Section~\ref{sec:bias}, we show that in the presence of errors-in-variables, estimated peer effects exhibit asymptotic biases and analyze how this bias depends on the interplay between characteristics and links. Section~\ref{sec:identification} provides formal conditions under which the linear-in-means model with errors-in-variables is identified. Based on these conditions, in Section~\ref{sec:estimation}, we propose GMM and 2SLS estimators that are straightforward to implement. A Monte Carlo simulation illustrates the applicability of our methods. Finally, Section~\ref{sec:conclusion} concludes. All proofs are in the Appendix. The Appendix and the Online Appendix contain additional results and extensions.

%%%%%%%%%%%%%%%%%%%%%%%%%%%%%%%%%%%%%%%%%%%%%%%%%%%%%%%%%%%%%%%%%%%%%%%
%%%%%%%%%%%%%%%%%%%%%%%%%%%%%%%%%%%%%%%%%%%%%%%%%%%%%%%%%%%%%%%%%%%%%%%
%%%%%%%%%%%%%%%%%%%%%%%%%%%%%%%%%%%%%%%%%%%%%%%%%%%%%%%%%%%%%%%%%%%%%%%
\section{Setup}\label{sec:setup}

A researcher observes data on outcomes, characteristics, and peers, and wants to estimate the impact of peers' characteristics and peers' outcomes on individual outcomes. We consider a data-generating process where a sequence $\{\boldy_s, \bx_s, \be_s, \bu_s, \bA_s\}_{s = 1,\dots, S}$ of $S$ i.i.d. network observations is drawn from a joint distribution. Network $s$ has size $N_s$, $\boldy_s$ is a $N_s \times 1$ vector of outcomes, $\bx_s$ is a $N_s \times 1$ vector of continuous characteristics, $\be_s$ is a $N_s \times 1$ vector of disturbances, $\bu_s$ is a $N_s \times 1$ vector of measurement errors, and $\bA_s$ is the $N_s \times N_s$ adjacency matrix of network $s$ where $(\bA_s)_{ij}=1$ if $j$ is a peer of $i$ and $0$ otherwise. Networks may be directed, i.e. $(\bA_s)_{ij}$ may differ from $(\bA_s)_{ji}$. The researcher observes outcomes $\boldy_s$, networks $\bA_s$ and mismeasured characteristics $\bnx_s = \bx_s + \bu_s$. 

We assume that the size of networks $N_s$ is uniformly bounded and consider many-network asymptotics. The number of observations $N = \sum_{s=1}^S N_s \rightarrow \infty$ as the number of networks $S \rightarrow \infty$. Throughout, all probability limits are with respect to $N$ and are assumed to exist and to be finite. Our setup resembles a cross-sectional setting where the researcher observes a large number of small networks (e.g., peers in classrooms or neighborhoods) at a single point in time. The assumption that outcomes, characteristics, and networks are jointly determined is fairly general. This notably covers setups with fixed networks as well as stochastic models of network formation. 

Given two random variables $a$ and $b$, we let $\expectation[a] := \plim \frac{1}{N} \sum_s \sum_i a_{si}$ denote the expectation of $a$, $\variance(a) := \expectation[a^2] - \expectation[a]^2$ its variance, and $\covariance(a, b) := \plim \frac{1}{N} \sum_s \sum_i a_{si} b_{si}-\expectation[a]\expectation[b]$ the covariance between $a$ and $b$.

\subsection{Linear-in-means model}\label{subsec:model}
The neighborhood of individual $i$, $\mathcal{N}_{si}$, is the set of $i$'s peers, $j \in \mathcal{N}_{si} \Longleftrightarrow (\bA_s)_{ij}=1$. We assume that every individual has at least one peer, $\forall i, \mathcal{N}_{si} \neq \emptyset$. The degree of individual $i$, $d_i$, is the number of $i$'s peers, $d_i=|\mathcal{N}_{si}|=\sum_j (\bA_s)_{ij}\geq 1$. Introduce the \textit{interaction} matrix, $\bG_s$, as $(\bG_s)_{ij} = (\bA_s)_{ij}/d_i$. This matrix is row-normalized---every row sums to one---and captures linear-in-means interactions. 

As our {baseline specification}, we consider the standard linear-in-means model of social interactions,
\begin{equation*}
    \begin{split}
        y_{si} &=  \alpha  + \gamma x_{si} + \delta \frac{1}{d_i}\sum_{j \in \mathcal{N}_{si}} x_{sj} + \beta \frac{1}{d_i}\sum_{j \in \mathcal{N}_{si}} y_{sj} + e_{si}, \\
        \nx_{si} &= x_{si} + u_{si},
    \end{split}
\end{equation*}
in which individual $i$'s outcome depends on her individual characteristic (captured by $\gamma$), her peers' average characteristic (captured by $\delta$), and her peers' average outcome (captured by $\beta$). We assume that $|\beta| < 1$ such that the reduced form is well-defined. The disturbance in the outcome equation is assumed to satisfy the standard conditional mean independence condition: i.e., $\expectation(e_{si} \mid \bx_s, \bG_s, \bu_s) = 0$. Measurement error arises because the researcher only observes a noisy proxy $\nx_{si}$ for the true characteristic $x_{si}$. 

Stacking observations, this model can be written compactly in matrix notation as 
\begin{subequations}\label{eq:model_baseline}
\begin{align}
\boldy_s &=  \alpha \mathbf{1}  + \gamma \bx_s + \delta \bG_s{\bx}_s  + \beta \bG_s{\boldy}_s+\be_s,\label{eq:modeloutcome} \\
\bnx_s &= \bx_s + \bu_s. \label{eq:modelerror}
\end{align}
\end{subequations}
Regressors without measurement error can be partialled out first using the Frisch-Waugh-Lovell (FWL) theorem.

To make the problem tractable, we impose some structure on the moments of the measurement error. Throughout, we assume that its conditional mean is unrelated to all other variables in the model and to the network structure.
\begin{assumption}\label{ass:measurementerror1}
    The measurement errors satisfy:
    \begin{equation*}
        \begin{split}
             \expectation[u_{si} \mid \bx_s, \boldG_s, \be_s] &= 0.
        \end{split}
    \end{equation*}
\end{assumption}
\noindent Despite its mildness, Assumption~\ref{ass:measurementerror1} alone can yield identification from first moments; Theorem~\ref{theorem:identification} states the exact conditions. To exploit second moments, we can additionally restrict the conditional variance and pairwise covariance of the measurement error.
\begin{assumption}\label{ass:measurementerror2}
    The measurement errors satisfy:
\[
\begin{alignedat}{2}
\mathbb{E}\!\left[u_{si}^{2}\mid \bx_s,\boldG_s,\be_s\right] &= \sigma_u^{2}, &\qquad
\mathbb{E}\!\left[u_{si}u_{sj}\mid \bx_s,\boldG_s,\be_s\right] &= 0, \forall i \neq j.
\end{alignedat}
\]
\end{assumption}
\noindent Importantly, we make no identifying assumptions on higher moments of the variables in the model, nor do we require full independence of the measurement error or the disturbance in the outcome equation. We also do not make distributional assumptions such as normality. 

In Online Appendix~\ref{app:extensions} we develop two extensions of the baseline model. First, we incorporate network fixed effects. This relaxes the conditional mean independence assumption and is especially useful when unobserved factors cluster within networks. Second, we allow for measurement error in outcomes. In a linear-in-means setting this is nontrivial, since outcomes enter both sides of the equation. We show that, under mild conditions, our identification and estimation results extend to these cases.

\subsection{Naive 2SLS estimator}\label{subsec:OLSestimator}
Because individuals' outcomes appear on both sides of \eqref{eq:modeloutcome}, OLS estimation delivers biased estimates. Researchers therefore typically resort to a 2SLS estimator where the endogenous average outcomes of peers are instrumented with an instrumental variable $z$. A popular, model-based choice is the average characteristic of peers at distance two \cite{bramoulle09}: $z_{si} := \sum_{j : d(i,j) = 2} \widetilde{x}_{sj} /  \sum_{j : d(i,j) = 2} 1$, where $d(i,j)$ denotes the network distance between $i$ and $j$.\footnote{Characteristics of peers at distance two are commonly used as instruments in applied work on peer effects (see, e.g., \citeNP{PATACCHINI20146, nicoletti, degiorgi19}).} This is our leading instrument, and the one for which we characterize the asymptotic bias in the next section.

Define $\boldy := \vect(\boldy_1, \boldy_2, \dots ,\boldy_S)$, $\bnx := \vect(\bnx_1, \bnx_2, \dots ,\bnx_S)$, $\be := \vect(\be_1, \be_2, \dots ,\be_S)$, $\bu := \vect(\bu_1, \bu_2$ $, \dots ,\bu_S)$, $\mathbf{z} := \vect(\mathbf{z}_1, \mathbf{z}_2, \dots, \mathbf{z}_S)$, and let $\bG := \diag(\bG_1, \bG_2, \dots, \bG_S)$ be the block-diagonal matrix that contains the interaction matrices. The \emph{naive} 2SLS estimator for the baseline specification can then be written as
\begin{equation}\label{eq:IV}
    \begin{bmatrix}
    \alphaIV & \gammaIV & \deltaIV & \betaIV 
    \end{bmatrix}^\intercal := \left(\mathbf{Z}^{\intercal}\mathbf{X}\right)^{-1}\mathbf{Z}^{\intercal} \boldy,
\end{equation}
where $\mathbf{X} := \begin{bmatrix}
    \mathbf{1} & \bnx & \bG{\widetilde{\bx}} & \bG{\boldy}
\end{bmatrix}$ and $\mathbf{Z} := \begin{bmatrix} \mathbf{1} & \bnx & \bG{\widetilde{\bx}} & \mathbf{z} \end{bmatrix}$.

In the presence of measurement error, and as shown below, the 2SLS estimator~\eqref{eq:IV} generally exhibits an asymptotic bias. The reason is that both the individual observed characteristic and average observed characteristics of peers are endogenous. To see why, substitute Equation~\eqref{eq:modelerror} in~\eqref{eq:modeloutcome}, which gives that 
\begin{equation}\label{eq:fullmeasuremnt}
    \boldy_s = \alpha \mathbf{1} + \gamma \bnx_s + \delta \bG_s {\bnx}_s + \beta  \bG_s{\boldy}_s + \underbrace{\mathbf{e}_s - \gamma \bu_s - \delta \bG_s{\bu}_s}_{=: \boldsymbol{\eta}_s}.
\end{equation}
The endogeneity problem arises because individual and peer characteristics are typically correlated with the composite error term: i.e.,  $\covariance(\widetilde{x}, \eta) \neq 0$ and $\covariance(G{\widetilde{x}}, \eta) \neq 0$ if $\gamma, \delta \neq 0$. At the core of the problem is an underidentification problem: there are three endogenous variables, but only one instrumental variable.

%%%%%%%%%%%%%%%%%%%%%%%%%%%%%%%%%%%%%%%%%%%%%%%%%%%%%%%%%%%%%%%%%%%%%%%
%%%%%%%%%%%%%%%%%%%%%%%%%%%%%%%%%%%%%%%%%%%%%%%%%%%%%%%%%%%%%%%%%%%%%%%
%%%%%%%%%%%%%%%%%%%%%%%%%%%%%%%%%%%%%%%%%%%%%%%%%%%%%%%%%%%%%%%%%%%%%%%
\section{Asymptotic bias}\label{sec:bias}
We study the asymptotic bias in peer effect estimates in the presence of errors-in-variables. Our results highlight the role of the interplay between characteristics and links. We first present results for the general case, and then build intuition by focusing on two special cases that are relevant in applied work.

\paragraph{The general case.}
We derive analytical expressions for the probability limits of the naive 2SLS estimates of the individual effect $\gammaIV$, the contextual peer effect $\deltaIV$, and the endogenous peer effect $\betaIV$. To focus the discussion, throughout this section we take the instrument $z$ to be the average characteristic of peers at distance two, as introduced in the previous section.\footnote{The alternative instrument $G^2 \widetilde{x}$, which is valid in the absence of measurement error \cite{bramoulle09}, is not valid under errors-in-variables. It is built from walks of length two, whose endpoints may lie at distance zero or one: $(G^2)_{ii} = \sum_k g_{ik}g_{ki} > 0$ whenever $i$ has a reciprocated link, and $(G^2)_{ij} = \sum_k g_{ik}g_{kj} >0$ for a peer $j$ whenever $i$ and $j$ share a peer. Hence $\covariance(G^2 \widetilde{x}, \eta) \neq 0$ in general.}  By Assumptions~\ref{ass:measurementerror1} and \ref{ass:measurementerror2}, $z$ is  uncorrelated with measurement error at distance zero and one, $\covariance(z, u) = \covariance(z, G{u}) = 0$. In addition, $\covariance(z, e) = 0$, so the instrument is valid (for a formal discussion, see Proposition~\ref{prop:instrumentsdiffmatrix} and Corollary~\ref{cor:lagged} in Section~\ref{sec:IV}).

By substituting Equations~\eqref{eq:modeloutcome} and \eqref{eq:modelerror} in \eqref{eq:IV}, and applying the Frisch-Waugh-Lovell (FWL) theorem to partial out the intercept, we express the asymptotic biases of $\gammaIV$,  $\deltaIV$, and $\betaIV$ in terms of the true model parameters and the matrices $\mathbf{S}$ and $\boldsymbol{\Sigma}$, which collect (co)variances in characteristics and errors.

\begin{lemma}\label{lemma:bias}
    Suppose that Assumptions~\ref{ass:measurementerror1} and \ref{ass:measurementerror2} hold and that the matrix $(\mathbf{S} + \boldsymbol{\Sigma})$ is invertible, where \begin{subequations}\label{eq:covariancemat}
    \begin{align*}
       \boldS := \begin{bmatrix}
        \variance(x) & \covariance(x, G{x}) & \covariance(x, G{y})  \\
        \covariance(G{x}, x) & \variance(G{x}) & \covariance(G{x}, G{y}) \\
        \covariance(z, x) & \covariance(z, G{x}) & \covariance(z, G{y})
        \end{bmatrix}, \qquad
        \boldSigma := \begin{bmatrix}
        \sigma_u^2 & 0 & 0\\
        0 & h_0 \sigma_u^2 & 0 \\
        0 & 0 & 0 \\
        \end{bmatrix},
    \end{align*}
    and $h_{0} := \plim \frac{1}{N} \sum_{s} \sum_{i}\frac{1}{d_{si}}$.
\end{subequations}
    Then the 2SLS estimates of $\gamma$, $\delta$, $\beta$ converge in probability to
    \begin{equation*}
    \plim  
    \begin{bmatrix}
    \gammaIV \\ 
    \deltaIV \\ 
    \betaIV %
    \end{bmatrix}  =\underbrace{(\mathbf{S}+\boldsymbol{\Sigma} )^{-1}\mathbf{S}}_{=:\bM}
    \begin{bmatrix}
    \gamma \\ 
    \delta \\ 
    \beta%
    \end{bmatrix}.
    \end{equation*}    
\end{lemma}

The matrix $\boldsymbol{\Sigma}$ has a specific structure. First, under Assumption~\ref{ass:measurementerror2}, measurement error is homoscedastic and uncorrelated across peers, so that $\variance(u) = \sigma_u^2$ and $\covariance(u, G{u}) = 0$. Moreover, the measurement error in the instrument is uncorrelated with the measurement error in both the individual characteristic and the average peer characteristic. Second, due to the averaging across peers, measurement error in the average peer characteristic has smaller variance than measurement error in the individual characteristic. In particular, we have $\variance(G{u}) = h_{0} \sigma_u^2$, where $h_{0} \leq 1$ is the arithmetic mean of the inverse degrees in the network. This network statistic captures the overall extent of the averaging across peers. 

Note that in the absence of measurement error, $\boldSigma = \mathbf{0}$ and $\mathbf{M} = \bI$, so there is no asymptotic bias, as expected. In many settings, the presence of measurement error entails an attenuation bias toward zero, such that estimates are smaller than the corresponding true model parameters (e.g., see \citeNP{wansbeek2000} for a detailed discussion). As argued above, however, this need not hold for the peer-effect coefficients in the linear-in-means model.

In the presence of measurement error, Lemma~\ref{lemma:bias} implies that each estimate is a linear combination of the three true parameters. The off-diagonal elements of $\bM$ capture whether the (biased) estimate of one effect is picking up part of another effect. As noted before, we refer to this phenomenon as \emph{smearing}. It can give rise to an expansion bias, falsely suggesting the presence of peer effects when there are none. For example, if $M_{31} \neq 0$, then $\betaIV$ inherits a smearing component from $\gamma$ (and if $M_{32} \neq 0$ from $\delta$). As a result, researchers may falsely detect endogenous peer effects even when $\beta =0$, provided that $\gamma \neq 0$ or $\delta \neq 0$. 

We study general properties of the probability limits in the next result.
We find it useful to work with the following block decomposition of matrices $\mathbf{S}$ and $\boldsymbol{\Sigma}$:
\begin{equation*}
    \mathbf{S}=
\begin{bmatrix}
\mathbf{S}_{11} & \mathbf{s}_{12} \\ 
\mathbf{s}_{21} & s_{22}%
\end{bmatrix}, \qquad     \boldsymbol{\Sigma} =
\begin{bmatrix}
\boldSigma_{11} & \mathbf{0} \\ 
\mathbf{0}  & 0
\end{bmatrix},
\end{equation*}
where $\mathbf{S}_{11}$ and $\boldsymbol{\Sigma}_{11}$ are $2\times 2$  matrices and $\mathbf{s}_{12}$ and $\mathbf{s}_{21}^\intercal$ are $2 \times 1$ vectors. Define $\mathbf{K}$ as the Schur complement of $s_{22}$ in $\boldS+\boldsymbol{\Sigma}$,
i.e., $\mathbf{K} := \boldS_{11}+\boldsymbol{\Sigma}_{11}- \frac{1}{s_{22}} \mathbf{s}_{12} \mathbf{s}_{21}$.

\begin{proposition}\label{proposition:bias}
    Suppose that Assumptions~\ref{ass:measurementerror1} and \ref{ass:measurementerror2} hold, that the matrix $(\mathbf{S} + \boldsymbol{\Sigma})$ is invertible, and that $s_{22} \neq 0$. Then the 2SLS estimates of $\gamma$, $\delta$, $\beta$ converge in probability to
    \begin{equation*}
    \plim  
    \begin{bmatrix}
    \gammaIV \\ 
    \deltaIV \\ 
    \betaIV %
    \end{bmatrix}= 
    \begin{bmatrix}
    \bI-\mathbf{K}^{-1}\boldsymbol{\Sigma}_{11} & \mathbf{0} \\ 
    \frac{1}{s_{22}}\mathbf{s}_{21}\mathbf{K}^{-1}\boldsymbol{\Sigma}_{11} & 1%
    \end{bmatrix}\begin{bmatrix}
    \gamma \\ 
    \delta \\ 
    \beta%
    \end{bmatrix}.
    \end{equation*}
When characteristics are independent of links and i.i.d. across individuals, this further simplifies to
\begin{equation*}
\plim  
    \begin{bmatrix}
    \gammaIV \\ 
    \deltaIV \\ 
    \betaIV %
    \end{bmatrix} = \begin{bmatrix}
    \frac{\variance(x)}{\variance(x) + \sigma_u^2} & 0 & 0 \\
    0 &  \frac{\variance(Gx)}{\variance(Gx) + h_0 \sigma_u^2}  & 0 \\
    0 & 0 & 1
\end{bmatrix}\begin{bmatrix}
    \gamma \\ 
    \delta \\ 
    \beta%
    \end{bmatrix}.
\end{equation*}
\end{proposition}

Proposition~\ref{proposition:bias} shows that, under classical measurement error, peer effect estimates are generally biased due to smearing. The estimated contextual peer effect $\deltaIV$ picks up part of the individual effect, and the estimated endogenous peer effect $\betaIV$ picks up parts of the individual and contextual peer effects. Interestingly, there is no
smearing in the opposite direction: $\gammaIV$ and $\deltaIV$ do not pick up any part of the endogenous peer effect.

These smearing effects depend in complex ways on the relationship between characteristics and network structure and are therefore difficult to sign analytically. Indeed, we show below that even in the simpler model without contextual peer effects (i.e., with only endogenous peer effects), the smearing term can be negative or positive. Nevertheless, even in this general
setup we can make some qualitative observations. As shown in
Appendix~\ref{app:biasprop}, the probability limit of $\betaIV$ can be written as
\begin{equation*}
    \begin{split}
        \plim \betaIV &= \frac{\sigma_u^2 \varphi_{z,\widetilde{x}}}{D^{IV}_1}\gamma + \frac{h_0 \sigma_u^2 \varphi_{z, G\widetilde{x}}}{D^{IV}_1} \delta + \beta,
    \end{split}
\end{equation*}
where $D^{IV}_1 := \covariance(z, Gy) - \varphi_{z,\widetilde{x}} \covariance(x, Gy) - \varphi_{z,G\widetilde{x}}\covariance(Gx, Gy)$, and where $\varphi_{z,\widetilde{x}}$ and $\varphi_{z,G\widetilde{x}}$ denote
the population regression coefficients of $\widetilde{x}$ and $G\widetilde{x}$, respectively, from regressing $z$ on a constant, $\widetilde{x}$, and
$G\widetilde{x}$. This representation shows that the smearing of $\gamma$ and $\delta$ into $\betaIV$ is governed by the extent to which the observed individual and average peer characteristics are correlated with the instrument. In particular, $\gamma$ and $\delta$ do not smear into $\betaIV$ when $\varphi_{z,\widetilde{x}}=0$ and $\varphi_{z,G\widetilde{x}}=0$, respectively.\footnote{In Appendix~\ref{app:biasprop}, we show that $\covariance(z, x) = \covariance(z, Gx) = 0$ implies $\varphi_{z,\widetilde{x}}=\varphi_{z,G\widetilde{x}}=0$.} Intuitively, this highlights that correlation between the instrument and the individual and peer characteristics---which can naturally arise under homophily---is a key driver of the asymptotic bias in $\betaIV$. With measurement error, the estimated endogenous peer effect absorbs parts of the individual effect and of the contextual peer effect, because the instrument spuriously picks up variation in the true individual and average peer characteristic.

In addition, the asymptotic bias in $\betaIV$ tends to decrease with instrument strength, as captured by $\covariance(z,Gy)$. Ceteris paribus, stronger instruments reduce the noise-to-signal ratio and therefore mitigate the bias induced by errors-in-variables. This suggests that errors-in-variables may exacerbate the consequences of weak instruments.

Proposition~\ref{proposition:bias} also specializes the result to the case in which characteristics are independent of links and i.i.d. across individuals, which is particularly relevant in experimental settings.\footnote{Both conditions matter. With randomly assigned links but a network-level component in characteristics, as across heterogeneous schools or villages, $\covariance(x,Gx)>0$ and smearing persists. This is the mechanism behind the findings of \citeA{ammermueller_peer_2009}.} In this case, the smearing effects disappear and $\gammaIV$ and $\deltaIV$ exhibit attenuation bias. Interestingly, $\betaIV$ remains consistent.

\paragraph{Two particular cases.}
To build intuition and further illustrate these effects, we next discuss two particular cases. The linear-in-means model accommodates two channels of peer effects, but applied work often estimates simpler variants that include only one channel.\footnote{Models with contextual peer effects only are estimated in \citeA{carrelljole, lavy}. Models with endogenous peer effects only are estimated in \citeA{Gaviria, trogdon,  paul2024spatialautoregressivemodelmeasurement}.} This naturally raises the question of whether the insights underlying Proposition~\ref{proposition:bias} can be sharpened in these simpler models.

Suppose first that there is \emph{no endogenous peer effect} ($\beta =0$), and that this restriction is known by the researcher. In Online
Appendix~\ref{app2:biascontextual}, we derive the analogue of
Proposition~\ref{proposition:bias} for the restricted model estimated by OLS. We find that even in this
case, $\deltaOLS$ generally exhibits bias in the direction of $
\gamma$. More precisely, 
\begin{equation}\label{eq:biasnoendogenous}
\plim \deltaOLS =\frac{\sigma _{u}^{2}\covariance(x,G{x})}{D^{OLS}}\gamma + \frac{\variance(x)\variance(G{x})-\covariance(x,G{x})^{2} + \sigma_u^2 \variance(G{x})}{D^{OLS}} \delta,
\end{equation}
where $D^{OLS} := \variance(x)\variance(G{x})-\covariance(x,G{x})^{2} +h_{0}\sigma
_{u}^{2}\variance(x)+\sigma _{u}^{2}\variance(G{x})+h_{0}\sigma _{u}^{4}$. Hence, estimates of contextual peer effects can exhibit asymptotic bias even when endogenous peer effects are absent.

In contrast to the general case, the smearing term can be signed here using the sign of $\covariance(x,Gx)$, since $D^{OLS}>0$. In particular, expansion bias is positive when individual characteristics are positively correlated with the average characteristic of their peers. Such positive correlation arises naturally under homophily, when similar agents are more likely to be connected. The magnitude of the expansion bias is also decreasing in the average inverse degree $h_0$.

The mechanism underlying the bias is as follows. When $\covariance(x,Gx)>0$, $\deltaOLS$ will spuriously pick-up part of the variation in $x$, and this tendency is more pronounced when the measurement error in $G\widetilde{x}$ is smaller. This is because $G\widetilde{x}$ then becomes a better proxy for $x$. Since $G\widetilde{x}$ averages across peers, Lemma~\ref{lemma:bias} implies that the variance of its measurement error is proportional to $h_0$. Hence, a lower $h_0$ implies less measurement error and therefore more smearing: a larger share of the individual effect $\gamma$ is loaded onto the estimated contextual peer effect $\deltaOLS$.

Conversely, suppose next that there is \emph{no contextual peer effect} ($\delta =0$), and that, again, this restriction is known by the researcher. In this case, peers' characteristics $G{\widetilde{x}}$ provide a natural instrument for peers' outcomes $G{y}$. Computations in Online Appendix~\ref{app2:biasendogenous} show that under 2SLS estimation, $\betaIV$ has an
asymptotic bias in the direction of $\gamma $:
\begin{equation}\label{eq:biasnocontextual}
\plim \betaIV  =\frac{\sigma _{u}^{2}\varphi_{G\widetilde{x}, \widetilde{x}}}{D_2^{IV}}\gamma + \beta,
\end{equation}
where $D_2^{IV} := \covariance(G{x},G{y})-\varphi_{G\widetilde{x}, \widetilde{x}}\covariance(x,G{y})$, and where $\varphi_{G\widetilde{x}, \widetilde{x}} := \frac{\covariance(x, Gx)}{\variance(x) + \sigma_u^2}$ denotes the population regression coefficient on $\widetilde{x}$ from regressing $G\widetilde{x}$ on a constant and $\widetilde{x}$. Hence, estimates of endogenous peer effects can be asymptotically biased even when contextual peer effects are absent.

The smearing term is non-zero only when $\covariance(x,Gx)\neq 0$. As in the general case, the extent of smearing is driven by correlation between the instrument and the observed characteristics. The asymptotic bias tends to decrease with instrument strength, as captured by $\covariance(Gx,Gy)$. 

When both types of peer effects are present, these sources of bias get combined, leading to the asymptotic biases highlighted in Proposition~\ref{proposition:bias}.

%%%%%%%%%%%%%%%%%%%%%%%%%%%%%%%%%%%%%%%%%%%%%%%%%%%%%%%%%%%%%%%%%%%%%%%
%%%%%%%%%%%%%%%%%%%%%%%%%%%%%%%%%%%%%%%%%%%%%%%%%%%%%%%%%%%%%%%%%%%%%%%
%%%%%%%%%%%%%%%%%%%%%%%%%%%%%%%%%%%%%%%%%%%%%%%%%%%%%%%%%%%%%%%%%%%%%%%
\section{Identification}\label{sec:identification}
We now consider the identification of the linear-in-means model of peer effects, accounting for errors-in-variables. Our aim is to establish the conditions required for identification in this model. We demonstrate that the model exhibits generic identification properties, even in the absence of external information. In particular, inherent features of the network can serve as a valuable tool for mitigating measurement error and enabling the identification of peer effects.

We propose two complementary approaches to identify the parameters of interest $(\alpha, \beta, \gamma, \delta, \sigma_u^2)$, depending on how much structure is imposed on measurement error. The first strategy is based on conditional mean restrictions and only relies on Assumption \ref{ass:measurementerror1} (mean-zero). It identifies $(\alpha, \beta, \gamma, \delta)$ by exploiting variation in average individual characteristics across network positions. This strategy is valid when, for instance, agents with more friends tend to have a higher value for the characteristic. The second identification strategy is based on conditional covariance restrictions and relies on Assumptions \ref{ass:measurementerror1} and \ref{ass:measurementerror2} (homoscedasticity and uncorrelatedness). It identifies $(\beta, \gamma, \delta, \sigma_u^2)$ by exploiting variation in covariances of individual characteristics across pairs of network positions.\footnote{In the absence of information on social connections, \citeA{depaularasulsouza19} implicitly leverage similar covariances to identify both social interactions and peer effects. In contrast to their approach, our results do not require panel data.} We maintain the weak assumptions on the disturbance term in Equation~\eqref{eq:modeloutcome} and on the measurement error in Equation~\eqref{eq:modelerror}, and therefore focus exclusively on mean and covariance restrictions. Researchers willing to impose more structure could exploit additional moments for identification.\footnote{For example, \citeA{rose17} identifies a linear-in-means model from second moments of outcomes under homoscedastic disturbances.}

For identification purposes, we form a single pooled interaction matrix $\bG_0$ of dimensions $N_0 \times N_0$ by placing all networks in the support on the diagonal of a block–diagonal matrix. The setup is feasible because bounded network size implies a finite support of interaction matrices, even as we observe arbitrarily many networks. We likewise stack observed outcomes and characteristics into conformable vectors $\boldy_0$ and $\widetilde{\bx}_0$.

We say that two nodes $i_0$ and $j_0$ have symmetric network positions if they belong to the same network $\bA$ in $\bG_0$ and if there exists a permutation of nodes' labels, $\pi$, that maps one node into another while preserving links. Formally, $\pi(i_0)=j_0, \pi(j_0)=i_0$ and for all pairs of network nodes $k_0,l_0$, $A_{\pi(k_0)\pi(l_0)}=A_{k_0l_0}$. In this case, $i_0$ and $j_0$ are indistinguishable and hence have the same conditional distribution of characteristics and outcomes. We then obtain the number of \textit{unique network positions} by counting symmetric positions only once. For instance, a star with $N_0$ nodes has two unique positions, center and periphery, while a circle with $N_0$ nodes has a single unique network position. By contrast, all positions are generally unique in a large complex network.

\subsection{Mean restrictions}
To derive the conditional mean restrictions, note that using Equation~\eqref{eq:fullmeasuremnt} the reduced form can be written as 
\begin{equation*}
    \begin{split}
        \boldy_s &= (\bI - \beta \bG_s)^{-1} [\alpha \mathbf{1} + (\gamma \bI + \delta \bG_s) \bnx_s + \boldsymbol{\eta}_s].
    \end{split}
\end{equation*}
Stacking networks, taking expectations on both sides, and using Assumption~\ref{ass:measurementerror1}, we obtain
\begin{equation}\label{eq:restmean}
    \expectation[\boldy_0 \mid \bG_0] = \underbrace{(\bI - \beta \bG_0)^{-1} [\alpha \mathbf{1} + (\gamma \bI + \delta \bG_0) \bm]}_{=: \mathbf{r}(\bm, \bG_0; \boldsymbol{\theta}_1)},
\end{equation}
where $\bm := \expectation[\bnx_0 \mid \bG_0]$ collects the expected characteristic by network position and $\boldsymbol{\theta}_1 := (\alpha, \beta, \gamma, \delta)$. Since there are $N_0$ network positions in total, this yields a nonlinear system of $N_0$ equations in four unknowns. Except for knife-edge cases, most networks should then deliver a substantial degree of overidentification. This identification strategy exploits variation in individual characteristics across network positions. By taking expectations by network position, we need no restriction on the covariance structure of measurement error in the network, and peers can make correlated mistakes.   

We next derive the precise identification conditions. Say that the model is \emph{identified from conditional mean restrictions} if there do not exist $\boldsymbol{\theta}_1$, $\boldsymbol{\theta}'_1$ with $\boldsymbol{\theta}_1 \neq \boldsymbol{\theta}'_1$ such that $\mathbf{r}(\bm, \bG_0; \boldsymbol{\theta}_1) = \mathbf{r}(\bm, \bG_0; \boldsymbol{\theta}'_1)$. We require $\beta\gamma + \delta \neq 0$; if this condition fails, it corresponds to a knife-edge case in which contextual and endogenous peer effects exactly offset one another \cite{bramoulle09}. 

\begin{theorem}\label{theorem:identification}
    Suppose that Assumption~\ref{ass:measurementerror1} holds and that $\beta \gamma + \delta \neq 0$. The parameters $\alpha, \beta, \gamma, \delta$ are identified from conditional mean restrictions if and only if the vectors $\mathbf{1}$, $\bm$, $\bG_0 \bm$, $\bG_0^2 \bm$ are linearly independent.
\end{theorem}

Identification based on mean restrictions requires the matrix $[\mathbf{1}, \bm, \bG_0 \bm, \bG_0^2 \bm]$ to have full rank, that is, rank four. This condition captures restrictions on the network structure and on the joint distribution of characteristics and links. Identification can only hold if there is enough variation in conditional expectations $m_{i_0}$. Identification fails to hold, for instance, with homogeneous marginals, when for every pair $i_0,j_0$, $x_{i_0}$ and $x_{j_0}$ have the same distributions. In that case, $m_{i_0}=m_{j_0}$ and $\bm=\lambda \mathbf{1}$ for some scalar $\lambda$. This is the case, in particular, when characteristics are independent of links and i.i.d. across individuals. Identification also fails if $\mathbf{m}$ is constant within each network but varies across networks. Then $\mathbf{G}_0\mathbf{m} = \mathbf{m}$ because every block of $\mathbf{G}_0$ is row-stochastic, and the rank is at most two. Network-level sorting is therefore not sufficient: the dependence must operate at the level of positions, for instance when individuals with higher characteristics have more or better-connected peers. 

Regardless of the joint distribution, identification also fails to hold when the network has fewer than four unique network positions. This happens, for instance, if every network is a star with $N_0$ nodes (two unique network positions) or if every network is a line with five nodes (three unique network positions).

More generally, this rank condition is stronger than the condition in \citeA{bramoulle09}, who show identification in the linear-in-means model without measurement error when $\bI$, $\bG_0$, $\bG_0^2$ are linearly independent. Indeed, suppose that identification fails to hold in the absence of measurement error. Then there exist $\lambda_0,\lambda_1,\lambda_2$ not all equal to zero such that $\lambda_0\bI+\lambda_1\bG_0+\lambda_2\bG_0^2=0$. This implies that $\lambda_0\bm+\lambda_1\bG_0\bm+\lambda_2\bG_0^2\bm=0$ and hence the vectors $\mathbf{1}$, $\bm$, $\bG_0 \bm$, $\bG_0^2 \bm$ are linearly dependent. Identification conditions are naturally more demanding with measurement error than without.

When there is variation in conditional expectations, identification is related to the spectral properties of the interaction matrix $\bG_0$. More precisely, suppose that identification fails to hold. There are two cases. Either there exist $\lambda_0,\lambda_1$ such that $\bG_0 \bm=\lambda_0 \mathbf{1}+\lambda_1\bm$. Then if $\lambda_1 \neq 1$, we can see that $\bm+\frac{\lambda_0}{\lambda_1-1} \mathbf{1}$ is an eigenvector of $\bG_0$ for the eigenvalue $\lambda_1$. Or there exist $\lambda_0,\lambda_1, \lambda_2$ such that $\bG_0^2 \bm=\lambda_0 \mathbf{1}+\lambda_1\bm+\lambda_2\bG_0 \bm$. In that case, and if $\lambda_1+\lambda_2 \neq 1$, we can check that $\bm+\frac{\lambda_0}{\lambda_1+\lambda_2-1}\mathbf{1}$ is an eigenvector of the matrix $\bG_0^2-\lambda_2\bG_0$ for the eigenvalue $\lambda_1$. In either case, the vector of characteristic expectations conditional on network positions must be precisely related to some eigenvector of a simple polynomial function of $\bG_0$.\footnote{In graph theory, it is well known that structural symmetries in a network are often associated with higher algebraic multiplicities of the eigenvalues of its interaction matrix \cite{biggs1993}. At one extreme lies the complete graph, whose interaction matrix has an eigenvalue with algebraic multiplicity $N_0-1$. At the other extreme, Erdős--Rényi random graphs typically exhibit simple spectra with high probability. Informally, the more ``asymmetric'' the network, the less likely it is that an arbitrary vector is an eigenvector of the associated interaction matrix.}

This suggests that non-identification based on mean restrictions is uncommon when $\bm$ is unrestricted and the network is sufficiently asymmetric. We next derive a sufficient condition for generic identification. The following proposition shows that, if the network is sufficiently rich, the set of mean vectors $\bm$ for which identification fails is negligible. To formally define what it means for identification to be generic, we must first account for the fact that two nodes $i$ and $j$ with symmetric network positions are indistinguishable, and notably satisfy $m_i=m_j$. Denote by $\check{N}_0$ the number of unique network positions, by $\check{\bm}$ the $\check{N}_0 \times 1$ vector of conditional expectations and by $\check{\bG}_0$ the $\check{N}_0 \times \check{N}_0$ interaction matrix defined over unique network positions.\footnote{$\check{\bm}$ is obtained from $\bm$ by keeping a unique entry for each unique network position. $\check{\bG}_0$ is obtained from $\bG_0$ by keeping a unique row for each unique network position and by summing the columns over symmetric network positions.} Note that Equation (\ref{eq:restmean}) and Theorem~\ref{theorem:identification} still hold when replacing $\bm$ by $\check{\bm}$ and $\bG_0$ by $\check{\bG}_0$.

Formally, for a fixed $\bG_0$, we say that the model is \emph{generically identified from conditional mean restrictions} if the set
\begin{equation*}
    \Big\{\check{\bm} \in \mathbb{R}^{\check{N}_0} :
\mathbf{r}(\check{\bm}, \check{\bG}_0; \boldsymbol{\theta}_1)
= \mathbf{r}(\check{\bm}, \check{\bG}_0; \boldsymbol{\theta}'_1)
\text{ for some } \boldsymbol{\theta}_1 \neq \boldsymbol{\theta}'_1 \Big\}
\end{equation*}
has Lebesgue measure zero. Our definition of generic identification is conceptually related to the order (or rank) conditions in systems of linear equations being generically satisfied \cite[886]{lewbel19}. Note, however, that non-generic cases may still be of substantive interest.

\begin{proposition}\label{proposition:genericity}
    Suppose that Assumption~\ref{ass:measurementerror1} holds and that $\beta \gamma + \delta \neq 0$.
    If $\check{\bG}_0$ has at least four distinct eigenvalues, the parameters $\alpha, \beta, \gamma, \delta$ are generically identified from conditional mean restrictions.    
\end{proposition}

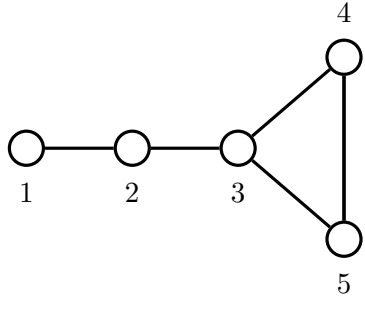
\begin{figure}[t]
\centering
\begin{tikzpicture}[
  node/.style={circle, draw, very thick, minimum size=4.5mm, inner sep=0pt},
  edge/.style={very thick},
  lab/.style={font=\small}
]

% Spacing (same as between 2 and 3)
\def\d{1.4}
\def\v{1.2}

% Nodes 1-2-3 on a horizontal line
\node[node] (n1) at (0,0) {};
\node[node] (n2) at (\d,0) {};
\node[node] (n3) at (2*\d,0) {};

% Nodes 4 (above) and 5 (below) to the right of 3 by the same distance d
\node[node] (n4) at (3*\d,\v) {};
\node[node] (n5) at (3*\d,-\v) {};

% Edges: 1-2-3-5-4 plus 3-4, and also 4-5
\draw[edge] (n1) -- (n2) -- (n3);
\draw[edge] (n3) -- (n5) -- (n4) -- (n3);
\draw[edge] (n4) -- (n5);

% Labels
\node[lab, below=2pt of n1] {$1$};
\node[lab, below=2pt of n2] {$2$};
\node[lab, below=2pt of n3] {$3$};
\node[lab, above=2pt of n4] {$4$};
\node[lab, below=2pt of n5] {$5$};

\end{tikzpicture}
\caption{A 5-node network in which positions 4 and 5 are symmetric}
    \label{fig:example}
\end{figure}

To illustrate Theorem \ref{theorem:identification} and Proposition \ref{proposition:genericity}, suppose that a researcher observes many networks of five nodes with the structure depicted in Figure~\ref{fig:example}. Here, positions 4 and 5 are symmetric, and hence the network structure has $\check{N}_0=4$ unique network positions. Parameters are identified from conditional mean restrictions if the following system of linear equations in $(\alpha,\beta,\gamma,\delta)$ has a unique solution:
\begin{align*}
\overline{y}_1&=\alpha+\gamma m_1+\delta m_2+\beta \overline{y}_2, \\
\overline{y}_2&=\alpha+\gamma m_2+\delta \left(\frac{1}{2}m_1+\frac{1}{2}m_3 \right)+\beta \left(\frac{1}{2}\overline{y}_1+\frac{1}{2}\overline{y}_3\right), \\
\overline{y}_3&=\alpha+\gamma m_3+\delta \left(\frac{1}{3}m_2+\frac{1}{3}m_4+\frac{1}{3}m_5\right)+\beta \left(\frac{1}{3}\overline{y}_2+\frac{1}{3}\overline{y}_4+\frac{1}{3}\overline{y}_5\right), \\
\overline{y}_4&=\alpha+\gamma m_4+\delta \left(\frac{1}{2}m_3+\frac{1}{2}m_5\right)+\beta \left(\frac{1}{2}\overline{y}_3+\frac{1}{2}\overline{y}_5\right), \\
\overline{y}_5&=\alpha+\gamma m_5+\delta \left(\frac{1}{2}m_3+\frac{1}{2}m_4\right)+\beta \left(\frac{1}{2}\overline{y}_3+\frac{1}{2}\overline{y}_4\right),
\end{align*}
where $\overline{y}_{i_0}$ denotes the expected outcome conditional on position $i_0$. By Theorem~\ref{theorem:identification}, provided $\beta\gamma+\delta \neq 0$, the parameters are identified if and only if $\mathbf{1}$, $\bm$, $\bG_0 \bm$, $\bG_0^2 \bm$ are linearly independent. Here,
\begin{equation*}
  \bG_0 = \begin{bmatrix}  
0 & 1 & 0 & 0 & 0 \\
\frac{1}{2} & 0 & \frac{1}{2} & 0 & 0 \\
0 & \frac{1}{3} & 0 &\frac{1}{3} & \frac{1}{3} \\
0 & 0 & \frac{1}{2} & 0 & \frac{1}{2} \\
0 & 0 & \frac{1}{2} & \frac{1}{2} & 0
\end{bmatrix},
\text{ and }\bG_0^2 = \begin{bmatrix}  
\frac{1}{2} & 0 & \frac{1}{2} & 0 & 0 \\
0 & \frac{2}{3} & 0 & \frac{1}{6} & \frac{1}{6} \\
\frac{1}{6} & 0 & \frac{1}{2} &\frac{1}{6} & \frac{1}{6} \\
0 & \frac{1}{6} & \frac{1}{4} & \frac{5}{12} & \frac{1}{6} \\
0 & \frac{1}{6} & \frac{1}{4} & \frac{1}{6} & \frac{5}{12}
\end{bmatrix},
\end{equation*}
and identification from mean restrictions holds if and only if the following matrix has rank four:
\begin{equation*}
[\mathbf{1}, \bm, \bG_0 \bm, \bG_0^2 \bm]=\begin{bmatrix}
    1 & m_1 & m_2 & \frac{1}{2}m_1+\frac{1}{2}m_3 \\
    1 & m_2 & \frac{1}{2}m_1+\frac{1}{2}m_3 & \frac{2}{3}m_2+\frac{1}{6}m_4+\frac{1}{6}m_5\\
    1 & m_3 & \frac{1}{3}m_2+\frac{1}{3}m_4+\frac{1}{3}m_5 & \frac{1}{6}m_1+\frac{1}{2}m_3+\frac{1}{6}m_4+\frac{1}{6}m_5 \\
    1 & m_4 & \frac{1}{2}m_3+\frac{1}{2}m_5 & \frac{1}{6}m_2+\frac{1}{4}m_3+\frac{5}{12}m_4+\frac{1}{6}m_5 \\
    1 & m_5 & \frac{1}{2}m_3+\frac{1}{2}m_4 & \frac{1}{6}m_2+\frac{1}{4}m_3+\frac{1}{6}m_4+\frac{5}{12}m_5 
\end{bmatrix}.
\end{equation*}
Then, to apply Proposition~\ref{proposition:genericity}, we merge the two symmetric positions 4 and 5. The interaction matrix defined over unique positions is now
\begin{equation*}
    \check{\bG}_0=\begin{bmatrix}
        0 & 1 & 0 & 0  \\
\frac{1}{2} & 0 & \frac{1}{2} & 0  \\
0 & \frac{1}{3} & 0 &\frac{2}{3} \\
0 & 0 & \frac{1}{2} & \frac{1}{2} \\
    \end{bmatrix}.
\end{equation*}
We can check that this matrix has four distinct eigenvalues. By Proposition~\ref{proposition:genericity}, we know that for this network structure, the set of $\check{\mathbf{m}} \in \mathbb{R}^4$ for which identification fails has Lebesgue measure zero.

\subsection{Covariance restrictions}
Another source of identification is provided by conditional covariance restrictions. Consider $\covariance(\boldy_0, \bnx_0 \mid \bG_0)$, the matrix that collects over all pairs $(i_0, j_0)$ of network positions the covariance between outcome at position $i_0$ and the observed characteristic at position $j_0$. Denote by $\bC := \variance(\bnx_0 \mid \bG_0 )$ the matrix of covariances of observed characteristics between network positions. Observe that under Assumption~\ref{ass:measurementerror1}, $\variance(\bnx_0 \mid \bG_0)=\variance(\bx_0 \mid \bG_0)+\variance(\bu_0 \mid \bG_0)$. If one is willing to additionally impose Assumption~\ref{ass:measurementerror2}, $\variance(\bu_0 \mid \bG_0)={\sigma_u^2} \bI$. Replacing $\boldy_0$ by its reduced-form expression gives
\begin{equation}\label{eq:restcovariance}
    \covariance(\boldy_0, \bnx_0 \mid \bG_0) = \underbrace{(\bI - \beta \bG_0)^{-1} (\gamma \bI + \delta \bG_0) (\bC - {\sigma_u^2} \bI)}_{=:\mathbf{R}(\bC, \bG_0; \boldsymbol{\theta}_2)},
\end{equation}
where $\boldsymbol{\theta}_2 := (\beta, \gamma, \delta, \sigma_u^2)$. This yields a nonlinear system of $N_0^2$ equations in four unknowns. As with mean restrictions, the model should then be substantially overidentified for most networks. This identification strategy exploits variation in covariances between observed characteristic across pairs of network positions. 

We derive a sufficient condition for identification in the following result. We say the model is \emph{identified from conditional covariance restrictions} if there do not exist $\boldsymbol{\theta}_2$, 
$\boldsymbol{\theta}'_2$ with $\boldsymbol{\theta}_2 \neq \boldsymbol{\theta}'_2$ such that $\mathbf{R}(\bC, \bG_0; \boldsymbol{\theta}_2) = \mathbf{R}(\bC, \bG_0; \boldsymbol{\theta}'_2)$.
\begin{theorem}\label{thm:identification2}
    Suppose that Assumptions~\ref{ass:measurementerror1} and \ref{ass:measurementerror2} hold and that $\beta \gamma + \delta \neq 0$. The parameters $\beta, \gamma, \delta, \sigma_u^2$ are identified from conditional covariance restrictions if the matrices $\bI$, $\bG_0$, $\bG_0^2$, $\bC$, $\bG_0\bC$, $\bG_0^2 \bC$ are linearly independent.
\end{theorem}

Let us illustrate how the identification conditions operate when characteristic covariance varies across pairs of network positions. Define, as before, \emph{network distance} $d(i_0,j_0)$ between two positions $i_0,j_0$ as the number of links in a shortest path between them and \emph{diameter} as the largest finite network distance. Consider a setup where covariance between individuals' characteristics varies with network distance, up to distance three, and where the diameter of the network is greater than or equal to five. Thus, $\mathbb{V}(x_{i_0} \mid \bG_0)=\sigma_x^2$, $\mathbb{C}(x_{i_0},x_{j_0}\mid \bG_0)=\rho_1 \sigma_x^2$ if $d(i_0,j_0)=1$, $\rho_2 \sigma_x^2$ if $d(i_0,j_0)=2$, $\rho_3 \sigma_x^2$ if $d(i_0,j_0)=3$ and $0$ if $d(i_0,j_0) \geq 4$ with $\rho_1,\rho_2,\rho_3>0$. In that case, we can show that the linear independence condition of Theorem \ref{thm:identification2} is satisfied. Indeed, there exists a pair of positions $i_0,j_0$ such that $d(i_0,j_0)=5$. For this pair, we see that $(\bG_0^2 \bC)_{i_0j_0}>0$ while $(\bG_0)_{i_0j_0}=(\bG_0^2)_{i_0j_0}=\bC_{i_0j_0}=(\bG_0 \bC)_{i_0j_0}=0$. Therefore, the matrix $\bG_0^2 \bC$ cannot be expressed as a linear combination of $\bI$, $\bG_0$, $\bG_0^2$, $\bC$, $\bG_0\bC$. Repeating the argument with a pair at distance four, we see that $\bG_0 \bC$ is not a linear combination of $\bI$, $\bG_0$, $\bG_0^2$, $\bC$, and with a pair at distance three, that $\bC$ is not a linear combination of $\bI$, $\bG_0$, $\bG_0^2$. Therefore, the matrices $\bI$, $\bG_0$, $\bG_0^2$, $\bC$, $\bG_0\bC$, $\bG_0^2 \bC$ are linearly independent and hence parameters $(\beta, \gamma, \delta, \sigma_u^2)$ are identified from covariance restrictions.

%%%%%%%%%%%%%%%%%%%%%%%%%%%%%%%%%%%%%%%%%%%%%%%%%%%%%%%%%%%%%%%%%%%%%%%
%%%%%%%%%%%%%%%%%%%%%%%%%%%%%%%%%%%%%%%%%%%%%%%%%%%%%%%%%%%%%%%%%%%%%%%
%%%%%%%%%%%%%%%%%%%%%%%%%%%%%%%%%%%%%%%%%%%%%%%%%%%%%%%%%%%%%%%%%%%%%%%
\section{Estimation}\label{sec:estimation}
While dependencies between characteristics and links may generate asymptotic bias, they simultaneously offer a valuable source of information that can be harnessed to develop consistent estimators for the parameters of interest. We next introduce straightforward and practical GMM and 2SLS estimators, directly built upon the conditional mean and covariance restrictions highlighted in the previous section.\footnote{While the GMM approach can yield more efficient parameter estimates, it does introduce nonlinearity into the estimation procedure. Online Appendix~\ref{app:fixedeffects} discusses how to apply the GMM estimators in the presence of network-specific fixed effects.}\textsuperscript{,}\footnote{In practice, researchers are advised to test for weak moments and weak instruments, since weak identification can lead to poor finite-sample properties and unreliable conventional inference.} Importantly, these estimators are applicable to setups with fixed as well as stochastic networks.

\subsection{GMM estimators}\label{sec:gmm}

\paragraph{Fixed networks.}
First consider the case where a fixed interaction matrix $\bG$ is observed repeatedly by the analyst.\footnote{With slight abuse of notation, throughout this subsection and the corresponding appendices, $\bG$ denotes a single interaction matrix, in contrast to the block-diagonal matrix defined in Section~\ref{subsec:OLSestimator}. Under fixed networks, $\bG_s = \bG$ for every $s$.} This case applies, for instance, in experimental setups where the researcher can control the network. 

We first introduce a GMM estimator based on the conditional mean restrictions in Equation~\eqref{eq:restmean}.\footnote{We construct the moment conditions from the \emph{structural} rather than the \emph{reduced} form of the model. These conditions are computationally more efficient because they avoid matrix inversion. Nevertheless, the two approaches are equivalent in the population and asymptotically equivalent when using the optimal weighting matrix.} The associated conditional moment conditions are $\expectation[\mathbf{v}_{1,s}(\boldsymbol{\theta}_1, \bG) \mid \bG] = \mathbf{0}$, where
\begin{equation}\label{eq:restmean2}
    \mathbf{v}_{1,s}(\boldsymbol{\theta}_1, \bG) := (\bI - \beta\bG)\boldy_s - \alpha \mathbf{1} - (\gamma \bI + \delta \bG) \bnx_s.
\end{equation}

\begin{proposition}\label{prop:GMMfirst}
    Suppose that Assumption~\ref{ass:measurementerror1} holds, that $\beta \gamma + \delta \neq 0$, and that the condition in Theorem~\ref{theorem:identification} is satisfied. Under standard regularity conditions (see Online Appendix~\ref{app:regularity}), the GMM estimator 
    \begin{equation*}
    \begin{split}
        \widehat{\boldsymbol{\theta}}^{GMM}_1 := \arg \min_{\boldsymbol{\theta}_1} \left(\frac{1}{S} \sum_{s}
            \mathbf{v}_{1,s}(\boldsymbol{\theta}_1, \bG) 
        \right)^\intercal \boldsymbol{\Omega}_1 \left(\frac{1}{S} \sum_{s} 
            \mathbf{v}_{1,s}(\boldsymbol{\theta}_1, \bG) 
        \right),
    \end{split}
\end{equation*}
delivers consistent parameter estimates (i.e., $\plim \widehat{\boldsymbol{\theta}}^{GMM}_1 = \boldsymbol{\theta}_1^*$, the true parameter value) for every positive definite weighting matrix $\boldsymbol{\Omega}_1$, and is asymptotically normal.
\end{proposition}

Given that our estimator treats networks as the unit of observation, and under the assumption that networks are independent and identically distributed, there is no need to adjust standard errors for cross-sectional dependence.

The GMM estimator from Proposition~\ref{prop:GMMfirst} simplifies to a closed-form generalized least squares (GLS) estimator. This position-level GLS estimator takes as inputs the position-specific network averages of outcomes, individual characteristics, and peer characteristics. We stress that the GLS interpretation is algebraic: the effective units of observation are network positions, so identification relies solely on variation in average characteristics across positions, not on individual-level variation within them.
\begin{corollary}\label{cor:GLS}
    Suppose that Assumption~\ref{ass:measurementerror1} holds, that $\beta \gamma + \delta \neq 0$, and that the condition in Theorem~\ref{theorem:identification} is satisfied. Then $\widehat{\boldsymbol{\theta}}^{GMM}_1$ is numerically equivalent to the position-level GLS estimator 
    \begin{equation*}
        \widehat{\boldsymbol{\theta}}^{GLS}_1 := \left(\overline{\mathbf{X}}^\intercal \boldsymbol{\Omega}_1 \overline{\mathbf{X}}\right)^{-1}\overline{\mathbf{X}}^\intercal \boldsymbol{\Omega}_1 \overline{\boldy},
    \end{equation*}
    where $\overline{\boldy} := \frac{1}{S}\sum_s \boldy_s$ and $\overline{\mathbf{X}} := \frac{1}{S}\sum_s \begin{bmatrix}
    \mathbf{1} & \bnx_s & \bG{\widetilde{\bx}}_s & \bG{\boldy}_s
\end{bmatrix}$.
\end{corollary}
When the weighting matrix is the identity, position-level GLS collapses to OLS on position-specific, network-averaged data: $\widehat{\boldsymbol{\theta}}^{OLS}_1 := \left(\overline{\mathbf{X}}^\intercal  \overline{\mathbf{X}}\right)^{-1}\overline{\mathbf{X}}^\intercal \overline{\boldy}$. Averaging at the position level eliminates the biases due to simultaneity and errors-in-variables, as both the disturbances and the measurement errors wash out within each position. This result is striking, not least because it does not hinge on the presence of errors-in-variables. It implies that, provided there is sufficient variation in average characteristics across network positions, the parameters of a linear-in-means model with contextual and endogenous peer effects can be consistently estimated by OLS on the position-specific averages.

We now introduce a GMM estimator that instead exploits the conditional covariance restrictions in Equation~\eqref{eq:restcovariance}. The associated conditional moment conditions are $\expectation[\mathbf{V}_{2,s}(\boldsymbol{\theta}_2, \bG) \mid \bG] = \mathbf{0}$, where
\begin{equation}\label{eq:restcovariance2}
    \mathbf{V}_{2,s}(\boldsymbol{\theta}_2, \bG) := (\bI - \beta \bG
    )\dot{\boldy}_s \dot{\widetilde{\bx}}^\intercal_s - (\gamma \bI + \delta \bG) \left(\dot{\widetilde{\bx}}_s\dot{\widetilde{\bx}}^\intercal_s - \sigma_u^2 \bI\right),
\end{equation}
with $\dot{\boldy}_s := \boldy_s - \expectation[\boldy_s \mid \bG]$ and $\dot{\widetilde{\bx}}_s := \widetilde{\bx}_s - \expectation[\widetilde{\bx}_s \mid \bG]$.\footnote{In
practice, the unknown conditional means are replaced by their sample averages, which affects neither consistency nor asymptotic normality.}
\begin{proposition}\label{prop:GMMsecond}
    Suppose that Assumptions~\ref{ass:measurementerror1} and \ref{ass:measurementerror2} hold, that $\beta \gamma + \delta \neq 0$, and that the condition in Theorem~\ref{thm:identification2} is satisfied. Let $\mathbf{v}_{2,s}(\boldsymbol{\theta}_2, \bG) := \vect (\mathbf{V}_{2,s}(\boldsymbol{\theta}_2, \bG))$. Under standard regularity conditions (see Online Appendix~\ref{app:regularity}), the GMM estimator
    \begin{equation*}
    \begin{split}
        \widehat{\boldsymbol{\theta}}^{GMM}_2 := \arg \min_{\boldsymbol{\theta}_2} \left(\frac{1}{S} \sum_{s}
            \mathbf{v}_{2,s}(\boldsymbol{\theta}_2, \bG) 
        \right)^\intercal \boldsymbol{\Omega}_2  \left(\frac{1}{S} \sum_{s} 
            \mathbf{v}_{2,s}(\boldsymbol{\theta}_2, \bG)
        \right),
    \end{split}
\end{equation*}
delivers consistent parameter estimates (i.e., $\plim \widehat{\boldsymbol{\theta}}^{GMM}_2 = \boldsymbol{\theta}_2^*$, the true parameter value) for every positive definite weighting matrix $\boldsymbol{\Omega}_2$, and is asymptotically normal.
\end{proposition}

\paragraph{Stochastic networks.} We now consider the case where networks are stochastic. In this setting, aggregating the moment conditions across network positions (mean restrictions) and pairs of network positions (covariance restrictions) is often desirable. First, aggregation reduces the dimensionality of the system and mitigates noise in the moments.\footnote{If a particular network position appears infrequently in the sample, the associated moments may be very noisy, so aggregation across positions is recommended. Moreover, for large or irregular networks, enumerating all positions can be computationally prohibitive.} Second, aggregation can improve comparability of moment conditions derived from heterogeneous network structures. 

Our approach is to convert the conditional moment restrictions into unconditional ones by taking linear combinations of the original moments. We start from the conditional moments
\begin{equation*}
    \expectation[\mathbf{v}_{m,s}(\boldsymbol{\theta}_m, \bG) \mid \bG] = \mathbf{0}, \quad m \in \{1,2\},
\end{equation*}
which hold for each network $\bG$. Let $N_G$ denote the number of individuals in $\bG$ and consider arbitrary aggregation matrices $\mathbf{H}_1(\bG)$ and $\mathbf{H}_{2}(\bG)$ of dimensions $N_G \times K_1$ and $N_G^2 \times K_2$, that may depend on the interaction matrix $\bG$ but not on individual characteristics or outcomes. By the law of iterated expectations, for each $m \in \{1,2\}$,
\begin{equation*}
    \begin{split}
        \mathbf{h}_m(\boldsymbol{\theta}_m) := \expectation[\mathbf{H}_m(\bG)^\intercal\mathbf{v}_{m,s}(\boldsymbol{\theta}_m, \bG)] &= \expectation\left[\mathbf{H}_m (\bG)^\intercal\expectation[\mathbf{v}_{m,s}(\boldsymbol{\theta}_m , \bG) \mid \bG]\right],
    \end{split}
\end{equation*}
so that $\mathbf{h}_m(\boldsymbol{\theta}_m^*) = \mathbf{0}$: a valid set of unconditional moment conditions for any such choice of $\mathbf{H}_m$.

Aggregation occurs as follows. Matrix $\mathbf{H}_1(\bG)$ first aggregates mean restrictions across positions within network $\bG$ into $K_1$ aggregated moments. These moments are then aggregated across networks. Similarly, matrix $\mathbf{H}_2(\bG)$ aggregates covariance restrictions across pairs of positions within network $\bG$ into $K_2$ aggregated moments, which are then aggregated across networks. Appendix~\ref{app:aggregation} gives formal conditions under which the parameters are identified from these aggregated moments.\footnote{Regarding efficiency, \citeA{CHAMBERLAIN1987305} and \citeA{NEWEY1993419} characterize optimal aggregation matrices, which depend on the inverse covariance matrix of the moment conditions and the Jacobian of those moments. In our setting, however, these objects are difficult to estimate reliably: they vary with the underlying network structure and are high dimensional, rendering their estimation impractical.} For the mean restrictions, a single rank condition is both necessary and sufficient for identification, and it can be met only with at least four such restrictions; for the covariance restrictions, an analogous rank condition remains sufficient, and calls for at least five restrictions. Naturally, the two sets of restrictions need not be used in isolation: mean and covariance restrictions can be pooled into a single system of unconditional moments, and this is the approach we take in our Monte Carlo simulation in Section~\ref{sec:simulation}.

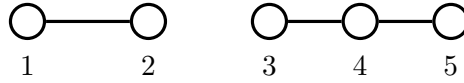
\begin{figure}[htbp]
\centering
\begin{tikzpicture}[
  node/.style={circle, draw, very thick, minimum size=4.5mm, inner sep=0pt},
  edge/.style={very thick},
  lab/.style={font=\small}
]

% --- (A) Dyadic pair ---
%\node[lab] at (-1.6,0.85) {(A) Dyad};
\node[node] (a1) at (-2.4,0) {};
\node[node] (a2) at (-0.8,0) {};
\draw[edge] (a1) -- (a2);
\node[lab, below=2pt of a1] {$1$};
\node[lab, below=2pt of a2] {$2$};

% --- (B) Three nodes on a line ---
%\node[lab] at (2.0,0.85) {(B) Line of 3};
\node[node] (b1) at (0.8,0) {};
\node[node] (b2) at (2.0,0) {};
\node[node] (b3) at (3.2,0) {};
\draw[edge] (b1) -- (b2) -- (b3);
\node[lab, below=2pt of b1] {$3$};
\node[lab, below=2pt of b2] {$4$};
\node[lab, below=2pt of b3] {$5$};

\end{tikzpicture}
\caption{A pair and a 3-node line}
\label{fig:twonetworks}
\end{figure}

To illustrate, suppose that a researcher observes many networks with the two structures depicted in Figure~\ref{fig:twonetworks}.  
In the population, the pair, denoted $\bG_p$, occurs with probability $\pi$, while the 3-node line, denoted $\bG_l$, occurs with probability $1-\pi$. The researcher aggregates mean restrictions by individuals' degree and covariance restrictions by the distance between pairs. To aggregate the mean restrictions within and across structures, the researcher specifies the aggregation matrices
\begin{equation*}
    \begin{split}
        \mathbf{H}_1(\bG_p) := \begin{bmatrix}
            1 & 1 \\
            0 & 0
        \end{bmatrix}^\intercal, \quad \mathbf{H}_1(\bG_l) := \begin{bmatrix}
            1 & 0 & 1 \\
            0 & 1 & 0
        \end{bmatrix}^\intercal.
    \end{split}
\end{equation*}
The aggregated moment conditions are then given by
\begin{equation*}
    \pi  \mathbf{H}_1(\bG_p)^\intercal \expectation[ \mathbf{v}_{1,s}(\boldsymbol{\theta}_1, \bG_p) \mid \bG_p]  + (1-\pi) \mathbf{H}_1(\bG_l)^\intercal  \expectation[\mathbf{v}_{1,s}(\boldsymbol{\theta}_1, \bG_l) \mid \bG_l] = \mathbf{0},
\end{equation*}
where the first and second rows collect mean restrictions for individuals with degrees one and two, respectively. This procedure computes two aggregated moments, one that sums the mean restrictions over all individuals in positions 1, 2, 3, and 5 (degree one) and another that sums the mean restrictions over all individuals in position 4 (degree two).

Similarly, to aggregate the covariance restrictions, the researcher specifies
\begin{equation*}
    \begin{split}
        \mathbf{H}_2(\bG_p) := \begin{bmatrix}
            1 & 0 & 0 & 1 \\
            0 & 1 & 1 & 0 \\
            0 & 0 & 0 & 0
        \end{bmatrix}^\intercal, \quad \mathbf{H}_2(\bG_l) := \begin{bmatrix}
            1 & 0 & 0 & 0 & 1 & 0 & 0 & 0 & 1 \\
            0 & 1 & 0 & 1 & 0 & 1 & 0 & 1 & 0 \\
            0 & 0 & 1 & 0 & 0 & 0 & 1 & 0 & 0 \\
        \end{bmatrix}^\intercal,
    \end{split}
\end{equation*}
where the first, second, and third rows collect covariance restrictions for pairs at distances zero, one, and two, respectively. The aggregated moment conditions are then
\begin{equation*}
    \pi  \mathbf{H}_2(\bG_p)^\intercal  \expectation[\mathbf{v}_{2,s}(\boldsymbol{\theta}_2, \bG_p) \mid \bG_p] + (1-\pi) \mathbf{H}_2(\bG_l)^\intercal  \expectation[\mathbf{v}_{2,s}(\boldsymbol{\theta}_2, \bG_l) \mid \bG_l] = \mathbf{0}.
\end{equation*}
This yields three aggregated moments: a first one that sums the variance restrictions over all individuals, a second one that sums covariance restrictions over pairs of individuals in positions (1,2), (3,4) and (4,5) (distance one), and a third one that sums covariance restrictions over the pair of individuals in positions (3,5) (distance two). Although the two network structures differ in size and share no common positions, the aggregation renders the restrictions comparable across graphs.

\subsection{2SLS estimators}\label{sec:IV}
The parameters of interest can also be recovered using a 2SLS estimator. This approach is simple to implement and readily accommodates network fixed effects and additional perfectly measured covariates. As discussed in Section~\ref{sec:setup}, measurement error renders the individual observed characteristic and average observed characteristics of peers endogenous, while the average outcome of peers is endogenous by construction. With three endogenous regressors $\big(\widetilde{x},\,G\widetilde{x},\,Gy\big)$, we require at least three instruments. We explore two approaches to constructing relevant and valid instruments from within the model: \emph{network-lagged characteristics} and \emph{network features}.\footnote{Alternatively, external information such as repeated measurements of the characteristics can be used to generate supplementary instruments. The three categories of instruments (network-lagged characteristics, network features, and external data) can be freely combined, provided that the rank condition is met.} The relevance of both rests on the interdependence between characteristics and links. 

Consider the modified 2SLS estimator
\begin{equation}\label{eq:IVmodified}
    \begin{bmatrix}
    \alphaIV & \gammaIV & \deltaIV & \betaIV
    \end{bmatrix}^\intercal
    := \left(\mathbf{X}^{\intercal} \mathbf{P}_Z \mathbf{X}\right)^{-1}
    \mathbf{X}^{\intercal} \mathbf{P}_Z\boldy,
\end{equation}
where $\mathbf{X}$ and $\boldy$ are as defined in Section~\ref{subsec:OLSestimator}, and $\mathbf{P}_Z := \mathbf{Z}\left(\mathbf{Z}^\intercal \mathbf{Z}\right)^{-1}\mathbf{Z}^\intercal$ denotes the projection matrix of the instrument matrix $\mathbf{Z}$. In contrast to the naive 2SLS estimator, $\mathbf{Z}$ contains only valid instruments. Using $\mathbf{P}_Z$ accommodates overidentification (i.e., $\mathbf{Z}$ may contain more than three instruments). The estimator is based on the unconditional moment conditions\footnote{With more than three excluded instruments, the sample analogues of all moments cannot be set to zero simultaneously; the 2SLS estimator in Equation~\eqref{eq:IVmodified} combines them with weighting matrix $(\mathbf{Z}^\intercal\mathbf{Z})^{-1}$.}
\begin{equation*}
    \expectation\!\left[\mathbf{Z}_s^\intercal
    (\boldy_s - \mathbf{X}_s\boldsymbol{\theta}_1)\right] = \mathbf{0},
\end{equation*}
where $\mathbf{X}_s$ and $\mathbf{Z}_s$ collect the rows of $\mathbf{X}$ and $\mathbf{Z}$ corresponding to network $s$. The parameters are identified from these moment conditions if and only if $\plim \frac{1}{N}\mathbf{Z}^\intercal\mathbf{X}$ has full column rank, which requires at least three excluded instruments. In practice, the rank condition can be checked with standard first-stage diagnostics.

We consider instruments stemming from two broad classes: (i) functions of the network alone, $\mathbf{z}_{s} := \mathbf{w}(\bG_s)$, and (ii) network-weighted characteristics, $\mathbf{z}_{s} := \bW(\bG_s)\bnx_s$. These classes cut across the two approaches introduced above: network-lagged characteristics fall in class (ii), while network features will generate instruments in both classes. The following proposition gives conditions for validity and connects each class to the conditional restrictions of Section~\ref{sec:identification}.

\begin{proposition}\label{prop:instrumentsdiffmatrix}
    Let $\mathbf{w}(\bG_s)$ and $\bW(\bG_s)$ depend on the network but not on characteristics or outcomes.
    \begin{enumerate}
        \item Suppose that Assumption~\ref{ass:measurementerror1} holds. Then $\mathbf{w}(\bG_s)$ is a valid instrument. In particular, for every $\boldsymbol{\theta}_1$, the moment condition implied by $\mathbf{w}(\bG_s)$ is a linear combination of the conditional mean restrictions, with weights that are fixed given the network.
        \item Suppose that Assumptions~\ref{ass:measurementerror1} and \ref{ass:measurementerror2} hold and that
        \begin{equation*}
            \expectation[\trace(\bW(\bG_s))]
            = \expectation[\trace(\bW(\bG_s)\bG_s^\intercal)] = 0.
        \end{equation*}
        Then $\bW(\bG_s)\bnx_s$ is a valid instrument. In particular, for every $\boldsymbol{\theta}_1$, the moment condition implied by $\bW(\bG_s)\bnx_s$ is a linear combination of the conditional mean and covariance restrictions, with weights that are fixed given the network.
    \end{enumerate}
\end{proposition}

Proposition~\ref{prop:instrumentsdiffmatrix} highlights that the 2SLS moments are linear combinations of the conditional mean and covariance restrictions, with network-determined weights. This has three implications. First, the instruments inherit their validity from the conditional restrictions, which vanish at the true parameters. Second, the instruments are jointly relevant (in the sense that $\plim \frac{1}{N}\mathbf{Z}^\intercal\mathbf{X}$ has full column rank) if the weights preserve sufficient identifying variation, paralleling the discussion in Section~\ref{sec:gmm}. For class (i), the rank condition is that of the aggregated mean restrictions in Appendix~\ref{app:aggregation}, with the vectors $\mathbf{w}$ as aggregation weights. In particular, no set of class (i) instruments satisfies the rank condition when the condition of Theorem~\ref{theorem:identification} fails, and the vectors $\mathbf{w}$ must not be collinear across positions. Class (ii) instruments additionally draw on the conditional covariances of characteristics, so $\bW$ should load on pairs whose characteristics are correlated. Third, estimates based on the 2SLS strategy are weakly less efficient than those of an efficiently weighted GMM strategy. Part of this gap could in principle be closed by enriching the instrument set, but part of it is irreducible within the class of instruments: the trace conditions of Proposition~\ref{prop:instrumentsdiffmatrix} force the weights of any such instrument to annihilate the terms $\gamma\sigma_u^2\bI$ and $\delta\sigma_u^2\bG$ in the covariance restrictions of Equation~\eqref{eq:restcovariance2}, since weights loading on these terms would correlate the instrument with the measurement error. These terms carry information about $\sigma_u^2$ as well as about $\gamma$ and $\delta$. Consequently, no choice of instruments within either class identifies $\sigma_u^2$ or recovers this part of the information about $\gamma$ and $\delta$, which contributes to the efficiency gap between 2SLS and GMM in the Monte Carlo simulation below.

\paragraph{Network-lagged characteristics.}
If the network is sufficiently sparse, one potential source of instrumental variables is network-lagged characteristics.\footnote{This approach is reminiscent of \citeA{grilicheshausman86}, who employ time-lagged observations in the context of linear panel data models.} The relevance and validity of these instruments rely on the correlation of individual characteristics across the network, under the assumption that measurement error is uncorrelated across individuals. When characteristics are correlated with links, the characteristics of peers at distances two or more can be used as instruments. 

\begin{corollary}\label{cor:lagged}
    Suppose that Assumptions~\ref{ass:measurementerror1} and \ref{ass:measurementerror2} hold. Let $\mathbf{A}_s^{(\tau)}$ denote the binary matrix indicating whether a pair within network $s$ is at distance $\tau$. Then $\bA_s^{(\tau)}\bnx_s$ is a valid instrument for every distance
    $\tau \geq 2$.\footnote{Rescaling the rows of $\bA_s^{(\tau)}$ by network-determined weights preserves its zero pattern and hence both trace conditions. In particular, the instrument $z_{si} = \sum_{j : d(i,j) = 2} \widetilde{x}_{sj} /  \sum_{j : d(i,j) = 2} 1$ of Sections~\ref{sec:setup} and \ref{sec:bias}, the row-normalized version of $\bA_s^{(2)}\bnx_s$, is valid.}
\end{corollary}

Importantly, when regressors are measured with error, powers of the interaction matrix $\bG$ do not yield valid transformations of individual characteristics, in sharp contrast to the error-free case \cite{bramoulle09}. Instruments based on higher-order powers become invalid because paths with endpoints at distances zero and one induce a correlation between the instrument and the error term. For instance, for a fixed network $\bG$ with $N_G$ individuals, 
\begin{equation*}
    \covariance(G^2 \widetilde{x}, \eta \mid \bG) = -\frac{\sigma_u^2}{N_G}\left(\gamma~\trace(\bG^2) + \delta~\trace(\bG^2 \bG^\intercal)\right),
\end{equation*}
which invalidates the instrument for most network structures whenever $\sigma_u^2 \neq 0$ and either $\gamma \neq 0$ or $\delta \neq 0$. For an undirected network, for example, $\trace(\bG^2) = \sum_{i,j} g_{ij} g_{ji} > 0$.

\paragraph{Network features.}
Another potential source of instruments comes from network structure. If individuals with higher values of the underlying characteristic tend to have more connections, then degree is a relevant and valid instrument; the same logic applies to other individual-level network statistics, such as local clustering, whenever these are correlated with the true characteristic. More generally, network structure can generate instruments that work through second-order moments, even when the associated statistic is uncorrelated with the characteristic in levels. Consistent with this idea, we propose recentered, model-based instruments reminiscent of \citeA{lewbel97}.

\begin{corollary}\label{cor:features}
    Let $\mathbf{t}_s := \mathbf{t}(\bG_s)$ be a network statistic
    that varies at the individual level.
    \begin{enumerate}
        \item Suppose that Assumption~\ref{ass:measurementerror1} holds. Then $\mathbf{t}_s$ is a valid instrument.
        \item Suppose that Assumptions~\ref{ass:measurementerror1} and \ref{ass:measurementerror2} hold. Then
    $(\mathbf{t}_s - \expectation[t]\mathbf{1})\circ\bnx_s$ and
    $\mathbf{d}_s\circ(\mathbf{t}_s - \expectation[t]\mathbf{1})\circ\bG_s\bnx_s$ are
    valid instruments, where $\circ$ denotes the entrywise product.
    \end{enumerate}
\end{corollary}

The first instrument belongs to class (i), being a function of the network alone; the second and third belong to class (ii), with the recentering ensuring that the trace conditions of Proposition~\ref{prop:instrumentsdiffmatrix} hold.\footnote{The population mean $\expectation[t]$ is the expected value of the statistic across all individuals in the population of networks. In practice, it is replaced by its sample counterpart, the pooled average $\bar{t} := \frac{1}{N}\sum_s \sum_i t_{si}$, under which the trace conditions hold exactly in every sample. This case is covered by Proposition~\ref{prop:instrumentsdiffmatrix}, which extends verbatim to weights that depend on the entire collection of networks $(\bG_1, \dots, \bG_S)$.} The three instruments draw on different parts of the conditional restrictions. The first uses mean restrictions only, while the second and third additionally place nonzero weights on covariance restrictions: the second on individuals' own-variance restrictions, the third on the covariance restrictions of linked pairs. 

Whether these instruments are relevant depends on how the network statistic covaries with the characteristics; Online Appendix~\ref{app:relevance} provides sufficient conditions for each instrument in Corollary~\ref{cor:features}. A statistic that is not relevant in the traditional sense may still be useful if it shifts second moments of the individual or average peer characteristic. To illustrate, consider an undirected network and an iso-correlational environment in which $\expectation[x_{i_0} \mid \bG_0] = 0$, $\variance(x_{i_0} \mid \bG_0) = \sigma_x^2$, and $\covariance(x_{i_0}, x_{j_0} \mid \bG_0) = \rho_{d({i_0}, {j_0})}\sigma_x^2$, so that pairwise covariances depend only on network distance. Then $t$ itself is not a relevant instrument, since $\covariance(t, x) = \covariance(t, G{x}) = 0$. However, as shown in Online Appendix~\ref{app:relevance}, the first-stage covariance of the recentered peer instrument equals
\begin{equation}\label{eq:variancescovariances}
    \covariance(d(t-\expectation[t])G\widetilde{x}, G\widetilde{x}) = \covariance(t, (d-1)c)\rho_1\sigma_x^2 + \covariance(t, (d-1)(1-c))\rho_2\sigma_x^2,
\end{equation}
where $c_{i_0}=\frac{\sum_{{j_0} \neq k_0 \in \mathcal{N}_{i_0}}A_{j_0k_0}}{d_{i_0}(d_{i_0}-1)}$ is the local clustering coefficient.\footnote{The local clustering coefficient measures how often an individual's friends are also friends with one another. It is the fraction of neighbor pairs of $i_0$ that are themselves linked (triangles involving $i_0$ over potential triangles).} If $\rho_1, \rho_2 \neq 0$, the instrument is therefore generically relevant whenever $\covariance(t, (d-1)c)\neq 0$ or $\covariance(t, (d-1)(1-c))\neq0$:
the statistic need only covary with the composition of peers' linked and
unlinked pairs. Because network statistics are often correlated in practice \cite{jacksonrogers07}, the relevance conditions are likely satisfied in many applications.

\subsection{Monte Carlo evidence}\label{sec:simulation}
We study the finite-sample performance of our GMM and 2SLS estimators using a Monte Carlo simulation. We consider a standard dyadic logit model of network formation in which links are conditionally independent across dyads. Despite its parsimony, this specification generates sufficient dependence between individual characteristics and network links to permit reliable recovery of the parameters of interest.

\paragraph{Setup.}
We consider a setup consisting of many small stochastic networks. In each simulation draw, we generate $500$ networks according to a dyadic model of link formation. The number of individuals within every network is fixed at $20$, a network size commonly encountered in empirical applications. The entire simulation procedure is repeated $500$ times. 

For each individual $i$ in network $s$, we draw characteristics according to
\begin{equation*}
    x_{si} = \xi_s + \xi_{si},
\end{equation*}
where $\xi_s \sim N(10, 1)$ captures network-level heterogeneity and $\xi_{si} \sim N(0, 2)$ captures individual-level variation. This structure mimics selection into networks based on observable characteristics.

Links between individuals within a given network are formed through a directed dyadic logit specification.\footnote{Logistic shocks are symmetric and reflect match-specific unobservables. If an individual is isolated, we assign that individual a random peer.} Specifically, the log odds of $i$ having an incoming link from $j$ in network $s$ is given by
\begin{equation*}
    \log \left(\frac{\Pr[a_{sij} = 1 \mid x_{si}, x_{sj}]}{\Pr[a_{sij} = 0 \mid x_{si}, x_{sj}]}\right) = \kappa + \mu x_{sj} - \nu|x_{si} - x_{sj}|,
\end{equation*}
where $\kappa$ is an intercept governing the overall propensity to form links, $\mu$ captures the extent to which individuals with higher values of the characteristic attract more connections, and $\nu$ controls the degree of homophily. Larger values of $\nu$ imply that individuals with more dissimilar characteristics are less likely to form links, on average. In our simulations, we set $(\kappa, \mu, \nu) = (-9, 1, 1)$.

Finally, we assume that measurement error for characteristics is drawn from a normal distribution: i.e., $u_{si} \sim N(0, 0.3)$. This implies a noise-to-signal ratio of 10\%. 

\paragraph{Moments and instruments.}
To illustrate the performance of our GMM approach, we pool the moment restrictions from Propositions~\ref{prop:GMMfirst} and \ref{prop:GMMsecond}. Following the aggregation strategy in Section~\ref{sec:gmm}, we group conditional mean restrictions by degree into bins of width three (i.e., 0--3, 3--6, 6--9, and 9+) and group conditional covariance restrictions by network distance (0,1, and 2+). For the modified 2SLS estimator in Equation~\eqref{eq:IVmodified}, we employ seven instruments to address the endogeneity of three regressors. In addition to average characteristics of individuals at network distance two, which are valid by Corollary~\ref{cor:lagged}, we construct six further instruments motivated by Corollary~\ref{cor:features}, based on degree and local clustering measures.\footnote{\label{footnote:instruments}Specifically, we construct the following instruments: $d$, $c$, $(d - \mathbb{E}[d])\nx$, $(c - \mathbb{E}[c])\nx$, $d(d - \mathbb{E}[d])G{\widetilde{x}}$, and $d(c - \mathbb{E}[c])G{\widetilde{x}}$, where $d$ denotes an individual's degree and $c$ their local clustering coefficient.} Although richer sets of moments and instruments could in principle be constructed, our objective is not to optimize their choice, but to assess the performance of the theoretically motivated ones introduced above.

We contrast our estimates to those obtained from the naive, and asymptotically biased, 2SLS estimator in Equation~\eqref{eq:IV}. This estimator is implemented by using only the average characteristics of individuals at network distance two as instrument. Because this specification employs a single instrument for three endogenous regressors, it is underidentified and therefore gives rise to an asymptotic bias, as discussed in Section~\ref{sec:bias}.

\paragraph{Results.}
Figure~\ref{fig:gmmiv} presents Monte Carlo evidence comparing the GMM and IV-based approaches. We set the model parameters to $(\alpha,\gamma, \delta, \beta) = (0,1,0.5,0.5)$. The solid line corresponds to the naive 2SLS estimator, which is biased due to the presence of errors-in-variables. The dotted line corresponds to the GMM estimator, while the dashed line corresponds to the modified 2SLS estimator.

The top panel reports estimates of the contextual peer effect~$\delta$. 
The naive 2SLS estimator exhibits substantial upward bias, with an average estimate of 1.03 (s.d. 0.11), whereas the GMM and modified 2SLS estimators are close to unbiased, with sample averages of 0.51 (s.d. 0.14) and 0.48 (s.d. 0.19), respectively. The bottom panel reports estimates of the endogenous peer effect~$\beta$. In this case, the naive 2SLS estimator is biased downward, with an average of 0.34 (s.d. 0.04), while the GMM and modified 2SLS estimators yield average estimates of 0.50 (s.d. 0.04) and 0.51 (s.d. 0.06), respectively. 

Overall, these results indicate that asymptotic bias in estimated peer effects is a first-order empirical concern. The naive estimate of the contextual peer effect is biased upward by more than 100\%, while the naive estimate of the endogenous peer effect is biased downward by more than 30\%. By contrast, our GMM and modified 2SLS estimators perform well in finite samples. Although the modified 2SLS estimator is more variable than the naive estimator---reflecting the treatment of three regressors as endogenous---its precision remains broadly comparable. The GMM estimator combines lower bias with higher precision.

\begin{figure}[ph!]
    \centering

    \begin{subfigure}{0.95\textwidth}
        \centering
        \includegraphics[width=\linewidth]{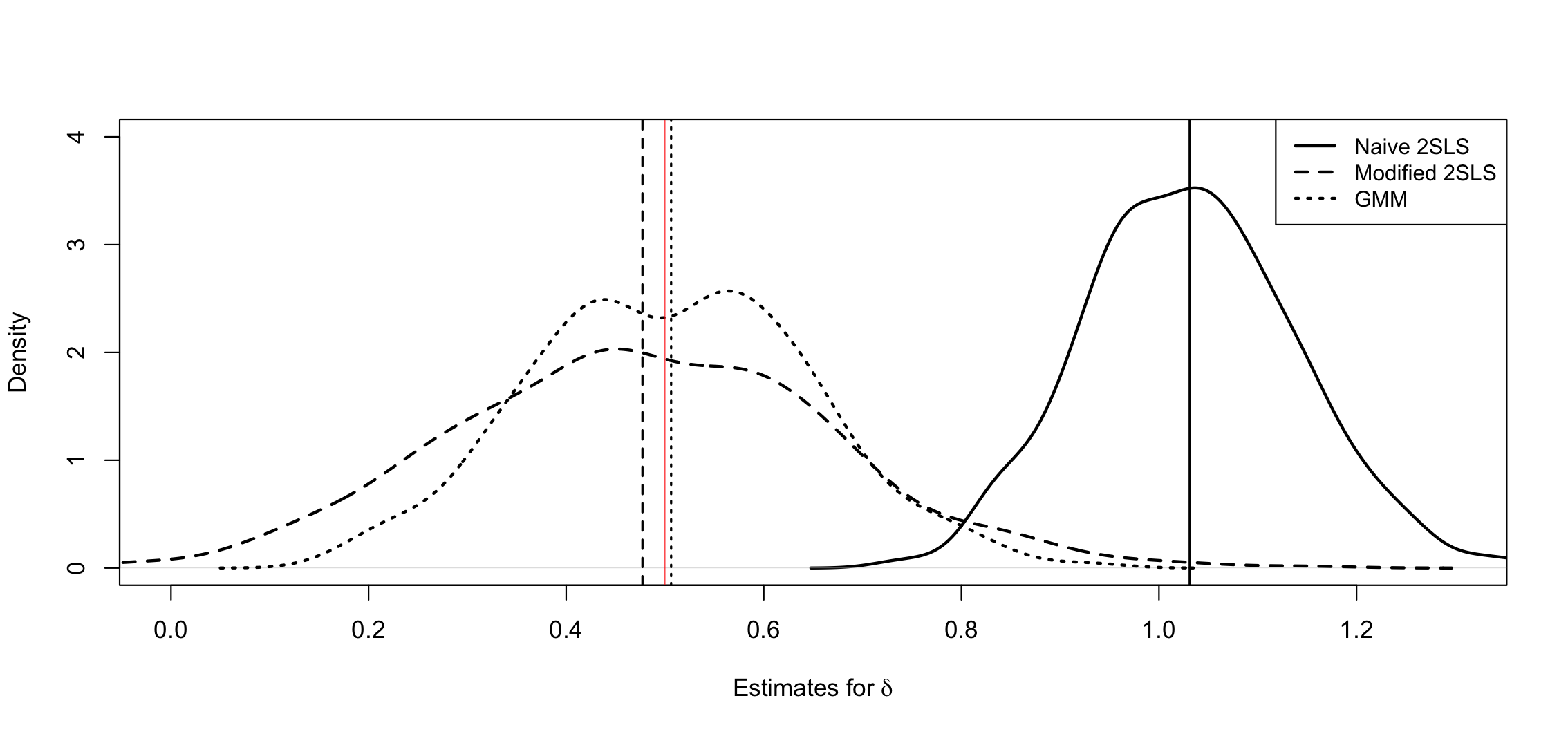}
        \caption{Contextual peer effect $\delta$}
        \label{fig:iv-beta}
    \end{subfigure}
    \hfill
    \begin{subfigure}{0.95\textwidth}
        \centering
        \includegraphics[width=\linewidth]{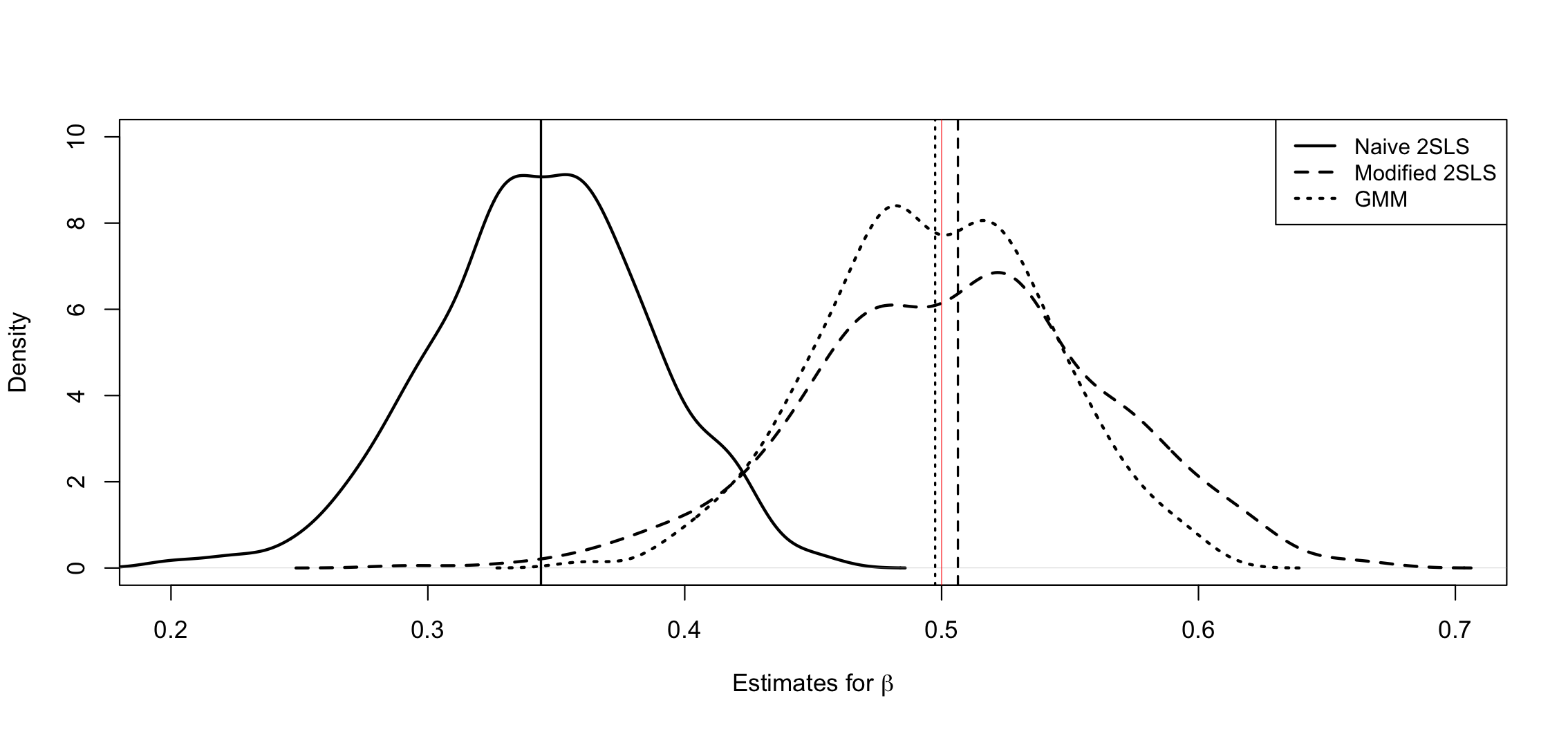}
        \caption{Endogenous peer effect $\beta$}
        \label{fig:iv-delta}
    \end{subfigure}

    \caption{Sampling distributions of the GMM and 2SLS estimators}
    \label{fig:gmmiv}

    \caption*{\footnotesize \textit{Notes:} Each panel displays kernel density estimates of Monte Carlo sampling distributions based on 500 replications for the GMM estimator and two 2SLS estimators. The GMM estimator follows Propositions~\ref{prop:GMMfirst} and \ref{prop:GMMsecond} and exploits both conditional mean and covariance restrictions; it is implemented using a two-step GMM procedure. The naive 2SLS estimator is implemented using Equation~\eqref{eq:IV}, while the modified 2SLS estimator is implemented using Equation~\eqref{eq:IVmodified}; the corresponding instrument set is described in Footnote~\ref{footnote:instruments}. The horizontal axis reports estimated parameter values and the vertical axis reports density. Black vertical lines indicate the sample means of each estimator. For both peer-effects parameters shown, the true parameter value equals 0.5 and is indicated by a thin red vertical line. In each replication, outcomes are generated on 500 networks with 20 individuals each. Simulations are conducted under parameter values $(\alpha, \gamma, \delta, \beta, \sigma_u^2, \kappa, \mu, \nu) = (0, 1, 0.5, 0.5, 0.3, -9, 1, 1)$, which imply an average network degree of 8.22 and an average local clustering coefficient of 0.60. The implied noise-to-signal ratio is 10\%.}
\end{figure}

%%%%%%%%%%%%%%%%%%%%%%%%%%%%%%%%%%%%%%%%%%%%%%%%%%%%%%%%%%%%%%%%%%%%%%%
%%%%%%%%%%%%%%%%%%%%%%%%%%%%%%%%%%%%%%%%%%%%%%%%%%%%%%%%%%%%%%%%%%%%%%%
%%%%%%%%%%%%%%%%%%%%%%%%%%%%%%%%%%%%%%%%%%%%%%%%%%%%%%%%%%%%%%%%%%%%%%%
%\section{Empirical application}

%%%%%%%%%%%%%%%%%%%%%%%%%%%%%%%%%%%%%%%%%%%%%%%%%%%%%%%%%%%%%%%%%%%%%%%
%%%%%%%%%%%%%%%%%%%%%%%%%%%%%%%%%%%%%%%%%%%%%%%%%%%%%%%%%%%%%%%%%%%%%%%
%%%%%%%%%%%%%%%%%%%%%%%%%%%%%%%%%%%%%%%%%%%%%%%%%%%%%%%%%%%%%%%%%%%%%%%
\section{Concluding remarks}\label{sec:conclusion}
This paper studies errors-in-variables in the linear-in-means model of social interactions. We show that classical measurement error generally induces asymptotic bias in naive 2SLS estimates of peer effects. While this bias is driven by the interplay between individual characteristics and network links, the same mechanism also offers an opportunity for identification and consistent estimation of the parameters of interest. We accordingly propose GMM and 2SLS estimators that are straightforward to implement, and demonstrate their performance in a Monte Carlo simulation.

Our analysis is a first step toward a general analysis of measurement error and peer effects in networks. We see three natural directions for future research. First, and given that the model is significantly overidentified, we believe that identification may be robust to relaxing Assumption~\ref{ass:measurementerror2}. It would be interesting to extend our analysis to setups where the variance-covariance structure of measurement errors may depend on the network. Second, it would be interesting, and challenging, to extend our analysis to settings where measurement error is non-classical. Recent research by \citeA{balestraeugsterpuljic23} has shown that when a binary characteristic is misclassified, an expansion bias can even arise under random group formation.\footnote{This additional 
expansion bias arises from the correlation between the average peer characteristic and the misclassification error in the individual characteristic.} We hypothesize that this additional bias will emerge in a network setting as well and that under non-random assignment, the usual asymptotic bias will reappear with qualitatively similar features to the bias studied here. Third, an exploration of the interplay between mismeasured characteristics and mismeasured links could contribute to a deeper and fuller understanding of the impact of measurement error on peer effect estimates. For instance, if homophily is correlated with friendship intensity, surveys that only sample a few best friends might deliver biased estimates of homophily. This, in turn, may interact with asymptotic bias due to mismeasured characteristics.

Our results also point to a natural next step for empirical work. While our simulations illustrate that the biases from mismeasured covariates can be sizable, their actual magnitude in practice depends on the data-generating process and will vary across applications. A systematic reassessment of earlier peer-effects studies through the lens of our framework, quantifying how sensitive their conclusions are to measurement error, would be a valuable exercise, which we leave for future research.

\clearpage
\bibliographystyle{apacite} 
\renewcommand\bibliographytypesize{\footnotesize} %Font size bibliography (\footnotesize or \small)
	\bibliography{references.bib}

%%%%%%%%%%%%%%%%%%%%%%%%%%%%%%%%%%%%%%%%
%%%%%%%%%%%%%%%%%%%%%%%%%%%%%%%%%%%%%%%%
%%%%%%%%%%%%%%%%%%%%%%%%%%%%%%%%%%%%%%%%
% APPENDIX
%%%%%%%%%%%%%%%%%%%%%%%%%%%%%%%%%%%%%%%%
%%%%%%%%%%%%%%%%%%%%%%%%%%%%%%%%%%%%%%%%
%%%%%%%%%%%%%%%%%%%%%%%%%%%%%%%%%%%%%%%%
\clearpage

\setcounter{page}{1}
\renewcommand{\thepage}{A\arabic{page}}

\begin{center}
    {\huge Appendix \\
    \vspace{10pt}
    \Large \textit{Measurement Error and Peer Effects in Networks}
  }
\end{center}

\appendix
\section{Proofs for results in the main text}\label{app:proofs}

\subsection{Proof of Lemma~\ref{lemma:bias}}
By the FWL theorem, we can partial out the intercept and work with demeaned variables. The resulting naive 2SLS estimator for $(\beta, \gamma,\delta)$ is
\begin{equation*}\label{eq:IVpartialled}
    \begin{bmatrix}
    \gammaIV & \deltaIV & \betaIV 
    \end{bmatrix}^\intercal = \left(\mathbf{Z}_w^{\intercal}\mathbf{X}_w\right)^{-1}\mathbf{Z}_w^{\intercal} \boldy_w,
\end{equation*}
where $\mathbf{X}_w := \begin{bmatrix}
    \bnx_w & (\bG{\widetilde{\bx}})_w & (\bG{\boldy})_w
\end{bmatrix}$, $\mathbf{Z}_w := \begin{bmatrix} \bnx_w & (\bG{\widetilde{\bx}})_w & \mathbf{z}_w \end{bmatrix}$, and, for any $N\times 1$ vector $\mathbf{a}$, the demeaned version is $\mathbf{a}_w := (\bI_N - \frac{1}{N}\bJ_N) \mathbf{a}$, with $\mathbf{I}_N$ the $N\times N$ identity and $\mathbf{J}_N$ the $N\times N$ matrix of ones. In this representation, all variables are expressed as deviations from their sample means.

Using the properties of probability limits, we have that
\begin{equation*}
    \begin{split}
        \plim \begin{bmatrix}
    \gammaIV & \deltaIV & \betaIV 
    \end{bmatrix}^\intercal &= \left(\plim \frac{1}{N} \mathbf{Z}_w^{\intercal}\mathbf{X}_w \right)^{-1} \left( \plim \frac{1}{N}\mathbf{Z}_w^{\intercal} \boldy_w \right),
    \end{split}
\end{equation*}
and under Assumptions~\ref{ass:measurementerror1} and \ref{ass:measurementerror2}, using that $\covariance(z, u) = \covariance(z, Gu) = 0$ since the instrument only contains characteristics of individuals at distance two, we obtain
\begin{equation*}
    \begin{split}
        \plim \frac{1}{N} \mathbf{Z}_w^{\intercal}\mathbf{X}_w  &= \plim \frac{1}{N}\begin{bmatrix}
        \bx^{\intercal}_w\bx_w & \bx^{\intercal}_w (\bG{\bx})_w &  \bx^{\intercal}_w (\bG{\boldy})_w \\
        (\bG{\bx})^{\intercal}_w \bx_w & (\bG{\bx})^{\intercal}_w(\bG{\bx})_w & (\bG{\bx})^{\intercal}_w(\bG{\boldy})_w \\
        {\mathbf{z}}^{\intercal}_w \bx_w & {\mathbf{z}}^{\intercal}_w(\bG{\bx})_w & {\mathbf{z}}^{\intercal}_w(\bG{\boldy})_w
        \end{bmatrix} \\
        &\qquad + \plim \frac{1}{N} \begin{bmatrix}
        \bu^{\intercal}_w\bu_w & \bu^{\intercal}_w (\bG{\bu})_w & 0\\
        (\bG{\bu})^{\intercal}_w \bu_w & (\bG{\bu})^{\intercal}_w (\bG{\bu})_w & 0 \\
        0 & 0 & 0 \\
        \end{bmatrix}, \\
              \plim \frac{1}{N}\mathbf{Z}_w^{\intercal} \boldy_w &=  \plim \frac{1}{N}\begin{bmatrix}
        \bx^{\intercal}_w\bx_w & \bx^{\intercal}_w (\bG{\bx})_w &  \bx^{\intercal}_w (\bG{\boldy})_w \\
        (\bG{\bx})^{\intercal}_w \bx_w & (\bG{\bx})^{\intercal}_w(\bG{\bx})_w & (\bG{\bx})^{\intercal}_w(\bG{\boldy})_w \\
        {\mathbf{z}}^{\intercal}_w \bx_w & {\mathbf{z}}^{\intercal}_w(\bG{\bx})_w & {\mathbf{z}}^{\intercal}_w(\bG{\boldy})_w
        \end{bmatrix} \begin{bmatrix}
            \gamma \\
            \delta \\ 
            \beta
        \end{bmatrix},
    \end{split}
\end{equation*}
where the second equality follows by substituting Equation~\eqref{eq:modeloutcome}. Evaluating these probability limits gives the desired result.

\subsection{Proof of Proposition~\ref{proposition:bias}}\label{app:biasprop}
\paragraph{The general case.}
Define $\mathbf{K}$ as the Schur complement of $s_{22}$ in $\boldS+\boldsymbol{\Sigma}$,
i.e., $\mathbf{K} := \boldS_{11}+\boldsymbol{\Sigma}_{11}- \frac{1}{s_{22}} \mathbf{s}_{12} \mathbf{s}_{21}$. Using the formula for the inverse of a block matrix, we obtain
\begin{equation*}
    (\boldS + \boldSigma)^{-1} = \begin{bmatrix}
        \mathbf{K}^{-1} & - \frac{1}{s_{22}}\mathbf{K}^{-1} \mathbf{s}_{12}  \\
        - \frac{1}{s_{22}}  \mathbf{s}_{21} \mathbf{K}^{-1} & \frac{1}{s_{22}}  + \frac{1}{s_{22}^2}  \mathbf{s}_{21} \mathbf{K}^{-1} \mathbf{s}_{12} 
    \end{bmatrix},
\end{equation*}
so that for $\mathbf{M} := (\boldS + \boldSigma)^{-1} \boldS$, we have
\begin{equation*}
    \begin{split}
        \mathbf{M} &= \begin{bmatrix}
        \mathbf{K}^{-1} & -  \frac{1}{s_{22}} \mathbf{K}^{-1} \mathbf{s}_{12} \\
        - \frac{1}{s_{22}} \mathbf{s}_{21} \mathbf{K}^{-1} & \frac{1}{s_{22}}  + \frac{1}{s_{22}^2}  \mathbf{s}_{21} \mathbf{K}^{-1} \mathbf{s}_{12} 
    \end{bmatrix} \begin{bmatrix}
\mathbf{S}_{11} & \mathbf{s}_{12} \\ 
\mathbf{s}_{21} & s_{22}%
\end{bmatrix} = \begin{bmatrix}
    \bI-\mathbf{K}^{-1}\boldsymbol{\Sigma}_{11} & \mathbf{0} \\ 
    \frac{1}{s_{22}} \mathbf{s}_{21}\mathbf{K}^{-1}\boldsymbol{\Sigma}_{11} & 1%
    \end{bmatrix}.
    \end{split}
\end{equation*}

Using the Sherman-Morrison formula, we obtain
\begin{equation*}
    \begin{split}
        \mathbf{K}^{-1} &= (\boldS_{11}+\boldsymbol{\Sigma}_{11})^{-1} + \frac{1}{s_{22}- \mathbf{s}_{21} (\boldS_{11}+\boldsymbol{\Sigma}_{11})^{-1} \mathbf{s}_{12}} (\boldS_{11}+\boldsymbol{\Sigma}_{11})^{-1} \mathbf{s}_{12}\mathbf{s}_{21} (\boldS_{11}+\boldsymbol{\Sigma}_{11})^{-1} \\
        &= (\boldS_{11}+\boldsymbol{\Sigma}_{11})^{-1} + \frac{1}{s_{22}- \boldsymbol{\varphi}_{z}^\intercal \mathbf{s}_{12}} (\boldS_{11}+\boldsymbol{\Sigma}_{11})^{-1} \mathbf{s}_{12}\boldsymbol{\varphi}_{z}^\intercal,
    \end{split}
\end{equation*}
where we used $\boldsymbol{\varphi}_{z}^\intercal := \mathbf{s}_{21} (\boldS_{11}+\boldsymbol{\Sigma}_{11})^{-1}$. This vector can be interpreted as the probability limit of the coefficients on $\widetilde{x}$ and $G\widetilde{x}$ obtained from regressing $z$ on a constant, $\widetilde{x}$, and $G\widetilde{x}$. Therefore,
\begin{equation*}
    \begin{split}
        \frac{1}{s_{22}}\mathbf{s}_{21}\mathbf{K}^{-1}\boldsymbol{\Sigma}_{11} &= \frac{1}{s_{22} - \boldsymbol{\varphi}_{z}^\intercal \mathbf{s}_{12}} \boldsymbol{\varphi}_{z}^\intercal \boldsymbol{\Sigma}_{11} \\
        &= \frac{1}{s_{22} - \boldsymbol{\varphi}_{z}^\intercal \mathbf{s}_{12}} \sigma_u^2 \begin{bmatrix}
            \varphi_{z,\widetilde{x}} &  h_0 \varphi_{z, G\widetilde{x}}
        \end{bmatrix}.
    \end{split}
\end{equation*}

If $\covariance(z, x) = \covariance(z, Gx) = 0$, then $\mathbf{s}_{21} = \mathbf{0}$, and hence $\varphi_{z,\widetilde{x}}=\varphi_{z,G\widetilde{x}}=0$.

\paragraph{Independent and i.i.d. characteristics.} Under independent and i.i.d. characteristics, we have $\mathbf{s}_{21} = \mathbf{0}$,
so that $\mathbf{K} = \boldS_{11}+\boldsymbol{\Sigma}_{11}$. In addition, $\mathbf{S}_{11} = \begin{bmatrix}
    \variance(x) & 0 \\
    0 & \variance(Gx)
\end{bmatrix}$. We therefore obtain
\begin{equation*}
    \begin{split}
        \bI-\mathbf{K}^{-1}\boldsymbol{\Sigma}_{11} &= \bI - \begin{bmatrix}
            \variance(x) + \sigma_u^2 & 0 \\
            0 & \variance(Gx) + h_0 \sigma_u^2
        \end{bmatrix}^{-1} \begin{bmatrix}
            \sigma_u^2 & 0 \\
            0 & h_0 \sigma_u^2
        \end{bmatrix} \\
        &= \bI - \begin{bmatrix}
            \frac{\sigma_u^2}{\variance(x) + \sigma_u^2} & 0 \\
            0 & \frac{h_0 \sigma_u^2}{\variance(Gx) + h_0\sigma_u^2}
        \end{bmatrix} \\
        &= \begin{bmatrix}
            \frac{\variance(x)}{\variance(x) + \sigma_u^2} & 0 \\
            0 & \frac{\variance(Gx)}{\variance(Gx) + h_0\sigma_u^2}
        \end{bmatrix}.
    \end{split}
\end{equation*}
Similarly, $\mathbf{s}_{21}=\mathbf{0}$ implies $\frac{1}{s_{22}} \mathbf{s}_{21}\mathbf{K}^{-1}\boldsymbol{\Sigma}_{11} = \mathbf{0}$.

\subsection{Proof of Theorem~\ref{theorem:identification}}
Let $\Theta_1 := \{ \boldsymbol{\theta}_1 : |\beta| < 1, \beta \gamma + \delta \neq 0\}$ denote the admissible parameter space. Suppose that $\boldsymbol{\theta}_1, \boldsymbol{\theta}'_1 \in \Theta_1$ generate the same conditional mean restriction, that is,
\begin{equation*}
    \begin{split}
        \expectation[\boldy_0 \mid \bG_0] &= (\bI - \beta \bG_0)^{-1} [\alpha \mathbf{1} + (\gamma \bI + \delta \bG_0) \bm], \\
            \expectation[\boldy_0 \mid \bG_0] &= (\bI - \beta' \bG_0)^{-1} [\alpha' \mathbf{1} + (\gamma' \bI + \delta' \bG_0) \bm], 
    \end{split}
\end{equation*}
or equivalently,
\begin{equation*}
    (\bI - \beta \bG_0)^{-1} [\alpha \mathbf{1} + (\gamma \bI + \delta \bG_0) \bm] = (\bI - \beta' \bG_0)^{-1} [\alpha' \mathbf{1} + (\gamma' \bI + \delta' \bG_0) \bm].
\end{equation*}
Left-multiplying both sides by $(\bI-\beta \bG_0)(\bI - \beta' \bG_0)$ yields
\begin{equation*}
    (\bI - \beta' \bG_0) [\alpha \mathbf{1} + (\gamma \bI + \delta \bG_0) \bm] = (\bI - \beta \bG_0) [\alpha' \mathbf{1} + (\gamma' \bI + \delta' \bG_0) \bm],
\end{equation*}
where we made use of the push-through identity $\bG_0 (\bI - \beta \bG_0)^{-1} = (\bI - \beta \bG_0)^{-1} \bG_0$. Expanding and rearranging gives $a_1 \mathbf{1} + a_2 \bm + a_3 \bG_0\bm + a_4\bG_0^2\bm = \mathbf{0}$, with coefficients
\[
\begin{alignedat}{2}
a_1&=\alpha(1-\beta')-\alpha'(1-\beta),&\qquad a_2&=\gamma-\gamma',\\
a_3&=\delta-\beta'\gamma+\beta\gamma'-\delta',&\qquad a_4&=-\beta'\delta+\beta\delta'.
\end{alignedat}
\]
Collect these coefficients in the vector $A_1(\boldsymbol{\theta}_1, \boldsymbol{\theta}'_1) := (a_1, a_2, a_3, a_4)$ and define $\mathcal{A}_1 := \{\mathbf{c} \in \mathbb{R}^4 : c_1 \mathbf{1} + c_2 \bm + c_3 \bG_0\bm + c_4\bG_0^2\bm = \mathbf{0}\}$. Since all steps above are equivalences, $\boldsymbol{\theta}_1$ and $\boldsymbol{\theta}'_1$ generate the same mean restriction if and only if $A_1(\boldsymbol{\theta}_1, \boldsymbol{\theta}'_1) \in \mathcal{A}_1$.

\paragraph{Sufficiency.} Suppose that $\mathbf{1}$, $\bm$, $\bG_0 \bm$, $\bG_0^2 \bm$ are linearly independent, so that $\mathcal{A}_1 = \{\mathbf{0}\}$ and hence $A_1(\boldsymbol{\theta}_1, \boldsymbol{\theta}'_1) = \mathbf{0}$. As in the proof of Proposition~1 in \citeA{bramoulle09}, $a_2 = a_3 = a_4 = 0$ together with $\beta \gamma + \delta \neq 0$ implies $(\beta, \gamma, \delta) = (\beta', \gamma', \delta')$. From $a_1 = 0$ and $|\beta| < 1$ it then follows that $\alpha = \alpha'$. Hence, $\boldsymbol{\theta}_1$ is identified.

\paragraph{Necessity.} Suppose instead that $\mathbf{1}$, $\bm$, $\bG_0 \bm$, $\bG_0^2 \bm$ are linearly dependent, so that $\mathcal{A}_1$ contains some $\mathbf{c} \neq \mathbf{0}$. Because $\mathcal{A}_1$ is a linear subspace, $t\mathbf{c} \in \mathcal{A}_1$ for every $t \in \mathbb{R}$. Fix $\boldsymbol{\theta}_1 \in \Theta_1$ and consider, for each $t$, the equation $A_1(\boldsymbol{\theta}_1, \boldsymbol{\theta}'_1) = t\mathbf{c}$ in $\boldsymbol{\theta}'_1$. As $A_1(\boldsymbol{\theta}_1, \cdot)$ is affine, this equation can be written as the linear system
\begin{equation*}
    \begin{bmatrix}
        -(1-\beta) & -\alpha & 0 & 0 \\
        0 & 0 & -1 & 0 \\
        0 & -\gamma & \beta  & - 1 \\
        0 & -\delta & 0 & \beta
    \end{bmatrix}
    \begin{bmatrix}
        \alpha' \\
        \beta' \\
        \gamma' \\
        \delta'
    \end{bmatrix} =
    \begin{bmatrix}
        tc_1 - \alpha \\
        tc_2 - \gamma \\
        tc_3 - \delta \\
        tc_4
    \end{bmatrix},
\end{equation*}
whose coefficient matrix has determinant $(1-\beta)(\beta\gamma + \delta)$, which is nonzero because $\boldsymbol{\theta}_1 \in \Theta_1$. For each $t$, the system therefore has a unique solution $\boldsymbol{\theta}'_1(t)$, which is affine, and in particular continuous, in $t$. This solution path has three properties. First, $\boldsymbol{\theta}'_1(0) = \boldsymbol{\theta}_1$, since $A_1(\boldsymbol{\theta}_1, \boldsymbol{\theta}_1) = \mathbf{0}$ and the solution at $t = 0$ is unique. Second, $\boldsymbol{\theta}'_1(t) \neq \boldsymbol{\theta}_1$ for all $t \neq 0$, since $A_1(\boldsymbol{\theta}_1, \boldsymbol{\theta}'_1(t)) = t\mathbf{c} \neq \mathbf{0} = A_1(\boldsymbol{\theta}_1, \boldsymbol{\theta}_1)$. Third, because $\Theta_1$ is open and $t \mapsto \boldsymbol{\theta}'_1(t)$ is continuous with $\boldsymbol{\theta}'_1(0) = \boldsymbol{\theta}_1 \in \Theta_1$, there exists $\varepsilon > 0$ such that $\boldsymbol{\theta}'_1(t) \in \Theta_1$ for all $|t| < \varepsilon$. Hence, for any $0 < |t| < \varepsilon$, the vector $\boldsymbol{\theta}'_1(t) \in \Theta_1$ is distinct from $\boldsymbol{\theta}_1$ yet generates the same mean restriction, so $\boldsymbol{\theta}_1$ is not identified.

\subsection{Proof of Proposition~\ref{proposition:genericity}}
We prove Proposition~\ref{proposition:genericity} through a series of lemmas. Let $\mathbf{T}({\check{\bm}}) := [\mathbf{1}, \check{\bm}, \check{\bG}_0 \check{\bm}, \check{\bG}_0^2 \check{\bm}]$ for $\check{\bm} \in \mathbb{C}^{\check{N}_0}$. 

\begin{lemma}\label{lem:genericity1}
    If $\check{\bG}_0$ has at least four distinct eigenvalues, there exists an $\check{\bm}^* \in \mathbb{C}^{\check{N}_0}$ such that $\mathrm{rank}~\mathbf{T}(\check{\bm}^*) = 4$. 
\end{lemma}
\begin{proof}
    Consider four distinct eigenvalues $\lambda_1, \lambda_2, \lambda_3, \lambda_4$ of $\check{\bG}_0$, and denote their corresponding eigenvectors by 
    $\mathbf{v}_1, \mathbf{v}_2, \mathbf{v}_3, \mathbf{v}_4 \in \mathbb{C}^{\check{N}_0}$. Since $\check{\bG}_0 \mathbf{1} = \mathbf{1}$, $\lambda_1 = 1$ is an eigenvalue with eigenvector $\mathbf{1}$. Because the eigenvalues  are pairwise distinct, their corresponding eigenvectors are linearly independent over $\mathbb{C}$.
    
    Define $\mathcal{V} := \mathrm{span}(\mathbf{v}_2, \mathbf{v}_3, \mathbf{v}_4)$ and pick $\check{\bm}^* = \mathbf{v}_2 + \mathbf{v}_3 + \mathbf{v}_4 \in \mathcal{V}$. Because $(\mathbf{v}_2, \mathbf{v}_3, \mathbf{v}_4)$ forms a basis of $\mathcal{V}$, we can rewrite the $\check{N}_0 \times 3$ matrix $[\check{\bm}^*, \check{\bG}_0\check{\bm}^*, \check{\bG}_0^2\check{\bm}^* ]$ in terms of this basis. In coordinates relative to $(\mathbf{v}_2, \mathbf{v}_3, \mathbf{v}_4)$ we have:\footnote{Note that $\check{\bG}_0\check{\bm}^* = \lambda_2 \mathbf{v}_2 + \lambda_3 \mathbf{v}_3 + \lambda_4 \mathbf{v}_4$ and $\check{\bG}_0^2 \check{\bm}^* = \lambda_2^2 \mathbf{v}_2 + \lambda_3^2 \mathbf{v}_3 + \lambda_4^2 \mathbf{v}_4$.}
    \begin{equation*}
        V := \begin{bmatrix}
            1 & \lambda_2 & \lambda_2^2 \\ 
            1 & \lambda_3 & \lambda_3^2 \\ 
            1 & \lambda_4 & \lambda_4^2 \\ 
        \end{bmatrix},
    \end{equation*}
    which is a Vandermonde matrix with $\det(V) = (\lambda_3 - \lambda_2)(\lambda_4 - \lambda_2)(\lambda_4 - \lambda_3) \neq 0$, since the eigenvalues are pairwise distinct. Therefore $\check{\bm}^*, \check{\bG}_0 \check{\bm}^*, \check{\bG}_0^2 \check{\bm}^*$ are linearly independent over $\mathbb{C}$ and $\mathrm{span}(\check{\bm}^*, \check{\bG}_0 \check{\bm}^*, \check{\bG}_0^2 \check{\bm}^*) = \mathcal{V}$. Since $\mathbf{v}_1 = \mathbf{1} \notin \mathcal{V}$ we have that $\mathrm{rank~} \mathbf{T}(\check{\bm}^*) = 4$. 
\end{proof}

\begin{lemma}\label{lem:genericity2}
    If there exists at least one $\check{\bm}^* \in \mathbb{C}^{\check{N}_0}$ for which $\mathrm{rank}~\mathbf{T}(\check{\bm}^*) = 4$ then the set of $\check{\bm} \in \mathbb{R}^{\check{N}_0}$ for which $\mathrm{rank}~\mathbf{T}(\check{\bm}) < 4$ has Lebesgue measure zero.
\end{lemma}
\begin{proof}
    Let $\Delta_1(\check{\bm}), \dots, \Delta_K(\check{\bm})$ be all $4 \times 4$ minors of $\mathbf{T}(\check{\bm})$.\footnote{That is, the determinants of every combination of four rows of $\mathbf{T}(\check{\bm})$.} Then $\mathrm{rank}~\mathbf{T}(\check{\bm}) < 4$ if and only if $\Delta_1(\check{\bm}) = \dots = \Delta_K(\check{\bm}) = 0$. Since $\mathrm{rank}~\mathbf{T}(\check{\bm}^*) = 4$, there exists a minor $\Delta_k$ such that $\Delta_k(\check{\bm}^*) \neq 0$. Because every entry of $\mathbf{T}(\check{\bm})$ is linear in $\check{\bm}$ with real coefficients, $\Delta_k$ is a polynomial in $(\check{m}_1, \dots, \check{m}_{\check{N}_0})$ with real coefficients. Since $\Delta_k(\check{\bm}^*) \neq 0$, not all of its coefficients are zero, and hence $\Delta_k$ does not vanish identically on $\mathbb{R}^{\check{N}_0}$. It is well known that the zero set of such a polynomial has Lebesgue measure zero. Since a rank deficiency implies that every minor vanishes, the set of $\check{\bm} \in \mathbb{R}^{\check{N}_0}$ with $\mathrm{rank}~\mathbf{T}(\check{\bm}) < 4$ is contained in the zero set of $\Delta_k$, and the result follows.
\end{proof}

Lemmas~\ref{lem:genericity1} and \ref{lem:genericity2} together prove the proposition, since $\mathrm{rank}~\mathbf{T}(\check{\bm}) = 4$ for $\check{\bm} \in \mathbb{R}^{\check{N}_0}$ is precisely the identification condition of Theorem~\ref{theorem:identification}.

\subsection{Proof of Theorem~\ref{thm:identification2}}
Let $\Theta_2 := \{ \boldsymbol{\theta}_2 : |\beta| < 1, \beta \gamma + \delta \neq 0\}$ denote the admissible parameter space. Suppose that $\boldsymbol{\theta}_2, \boldsymbol{\theta}'_2 \in \Theta_2$ generate the same conditional covariance restriction, that is,
\begin{equation*}
    \begin{split}
        \covariance(\boldy_0, \bnx_0 \mid \bG_0) &= (\bI - \beta \bG_0)^{-1} (\gamma \bI + \delta \bG_0) (\bC - \sigma_u^2 \bI), \\
        \covariance(\boldy_0, \bnx_0 \mid \bG_0) &= (\bI - \beta' \bG_0)^{-1} (\gamma' \bI + \delta' \bG_0) (\bC - {\sigma_u^2}' \bI),
    \end{split}
\end{equation*}
or equivalently, 
\begin{equation*}
    (\bI - \beta \bG_0)^{-1} (\gamma \bI + \delta \bG_0) (\bC - \sigma_u^2 \bI) = (\bI - \beta' \bG_0)^{-1} (\gamma' \bI + \delta' \bG_0) (\bC - {\sigma_u^2}' \bI).
\end{equation*}
Left-multiplying both sides by $(\bI-\beta \bG_0)(\bI - \beta' \bG_0)$ gives, using the push-through identity as in the proof of Theorem~\ref{theorem:identification},
\begin{equation*}
    (\bI - \beta' \bG_0) (\gamma \bI + \delta \bG_0) (\bC - \sigma_u^2 \bI) = (\bI - \beta \bG_0) (\gamma' \bI + \delta' \bG_0) (\bC - {\sigma_u^2}' \bI).
\end{equation*}
Expanding and rearranging gives $a_1 \bI + a_2 \bG_0 + a_3 \bG_0^2 + a_4 \bC + a_5 \bG_0 \bC + a_6 \bG_0^2 \bC = \mathbf{0}$, with coefficients
\[
\begin{alignedat}{2}
a_1&=-\gamma\sigma_u^2+\gamma'{\sigma_u^2}', &\qquad
a_2&=(\beta'\gamma-\delta)\sigma_u^2-(\beta\gamma'-\delta'){\sigma_u^2}',\\
a_3&=\beta'\delta\,\sigma_u^2-\beta\delta'\,{\sigma_u^2}', &\qquad
a_4&=\gamma-\gamma',\\
a_5&=-(\beta'\gamma-\delta)+\beta\gamma'-\delta', &\qquad
a_6&=-\beta'\delta+\beta\delta'.
\end{alignedat}
\]
Collect these coefficients in the vector $A_2(\boldsymbol{\theta}_2, \boldsymbol{\theta}'_2) := (a_1, a_2, a_3, a_4, a_5, a_6)$ and define $\mathcal{A}_2 := \{\mathbf{c} \in \mathbb{R}^6 : c_1 \bI + c_2 \bG_0 + c_3 \bG_0^2 + c_4 \bC + c_5 \bG_0 \bC + c_6 \bG_0^2 \bC = \mathbf{0}\}$. Since all steps above are equivalences, $\boldsymbol{\theta}_2$ and $\boldsymbol{\theta}'_2$ generate the same covariance restriction if and only if $A_2(\boldsymbol{\theta}_2, \boldsymbol{\theta}'_2) \in \mathcal{A}_2$.

Suppose that $\bI$, $\bG_0$, $\bG_0^2$, $\bC$, $\bG_0\bC$, $\bG_0^2 \bC$ are linearly independent, so that $\mathcal{A}_2 = \{\mathbf{0}\}$ and hence $A_2(\boldsymbol{\theta}_2, \boldsymbol{\theta}'_2) = \mathbf{0}$.  As in the proof of Proposition~1 in \citeA{bramoulle09}, $a_4 = a_5 = a_6 = 0$ together with $\beta \gamma + \delta \neq 0$ implies $(\beta, \gamma, \delta) = (\beta', \gamma', \delta')$. Substituting into $a_1$ and $a_2$ gives
\begin{equation*}
    a_1 = -\gamma\,(\sigma_u^2 - {\sigma_u^2}') = 0,
    \qquad
    a_2 = (\beta\gamma - \delta)\,(\sigma_u^2 - {\sigma_u^2}') = 0.
\end{equation*}
Suppose that $\sigma_u^2 \neq {\sigma_u^2}'$. Then $\gamma = 0$ and
$\beta\gamma - \delta = 0$, hence $\delta = 0$ and
$\beta\gamma + \delta = 0$, contradicting the maintained assumption.
Therefore $\sigma_u^2 = {\sigma_u^2}'$. Hence, $\boldsymbol{\theta}_2$ is identified.

\subsection{Proof of Corollary~\ref{cor:GLS}}
Since $\mathbf{v}_{1,s}(\boldsymbol{\theta}_1, \bG)$ is linear in $\boldsymbol{\theta}_1$, we can write $\mathbf{v}_{1,s}(\boldsymbol{\theta}_1, \bG) = \boldy_s - \mathbf{X}_s \boldsymbol{\theta}_1$, where $\mathbf{X}_s := \begin{bmatrix} \mathbf{1} & \bnx_s & \bG\widetilde{\bx}_s & \bG\boldy_s \end{bmatrix}$. The GMM objective in Proposition~\ref{prop:GMMfirst} is a convex quadratic in $\boldsymbol{\theta}_1$, so its minimizer is characterized by the first-order condition
\begin{equation*}
    \begin{split}
        \left(\frac{1}{S} \sum_{s}
            \frac{\partial \mathbf{v}_{1,s}(\boldsymbol{\theta}_1, \bG)}{\partial \boldsymbol{\theta}_1} 
        \right)^\intercal \boldsymbol{\Omega}_1 \left(\frac{1}{S} \sum_{s} 
            \mathbf{v}_{1,s}(\boldsymbol{\theta}_1, \bG) 
        \right) = \mathbf{0}.
    \end{split}
\end{equation*}
Substituting $\frac{\partial \mathbf{v}_{1,s}(\boldsymbol{\theta}_1, \bG)}{\partial \boldsymbol{\theta}_1} = -\mathbf{X}_s$ and rearranging, using $\overline{\mathbf{X}} = \frac{1}{S}\sum_s \mathbf{X}_s$ and $\overline{\boldy} = \frac{1}{S}\sum_s \boldy_s$, yields
\begin{equation*}
    \widehat{\boldsymbol{\theta}}^{GMM}_1 = \left(\overline{\mathbf{X}}^\intercal \boldsymbol{\Omega}_1 \overline{\mathbf{X}}\right)^{-1}\overline{\mathbf{X}}^\intercal \boldsymbol{\Omega}_1 \overline{\boldy},
\end{equation*}
provided $\overline{\mathbf{X}}^\intercal \boldsymbol{\Omega}_1 \overline{\mathbf{X}}$ is invertible. The latter is precisely the condition for the minimizer in Proposition~\ref{prop:GMMfirst} to be unique.

\subsection{Proof of Proposition~\ref{prop:instrumentsdiffmatrix}}
\paragraph{Instrument class (i): functions of the network.}  
Fix an arbitrary $\boldsymbol{\theta}_1$ and $\bG$. Since $\mathbf{w}(\bG)$ is nonrandom conditional on $\bG$, we have that
\begin{equation*}
    \begin{split}
    \expectation\!\left[\mathbf{z}_{s}^\intercal
    (\boldy_s - \mathbf{X}_s\boldsymbol{\theta}_1) \mid \bG \right]
    &= \expectation\!\left[\mathbf{w}(\bG)^\intercal\,
    \mathbf{v}_{1,s}(\boldsymbol{\theta}_1, \bG) \mid \bG \right] \\
    &= \mathbf{w}(\bG)^\intercal\,
    \expectation\!\left[\mathbf{v}_{1,s}(\boldsymbol{\theta}_1, \bG)
    \mid \bG \right].
    \end{split}
\end{equation*}
By the law of iterated expectations, averaging over the distribution of $\bG$ gives
\begin{equation*}
    \expectation\!\left[\mathbf{z}_{s}^\intercal
    (\boldy_s - \mathbf{X}_s\boldsymbol{\theta}_1) \right]
    = \expectation\!\left[\mathbf{w}(\bG)^\intercal\,
    \expectation\!\left[\mathbf{v}_{1,s}(\boldsymbol{\theta}_1, \bG)
    \mid \bG \right]\right].
\end{equation*}
The moment condition is therefore a linear combination of the conditional mean restrictions, with weights $\mathbf{w}(\bG)$ that are fixed given the network. Validity follows by evaluating at
$\boldsymbol{\theta}_1 = \boldsymbol{\theta}_1^{*}$: under
Assumption~\ref{ass:measurementerror1}, we have that
$\expectation[\mathbf{v}_{1,s}(\boldsymbol{\theta}_1^{*}, \bG) \mid
\bG] = \mathbf{0}$ for every $\bG$. Hence $\expectation\!\left[\mathbf{z}_{s}^\intercal
    (\boldy_s - \mathbf{X}_s\boldsymbol{\theta}^*_1) \right] = 0$.

\paragraph{Instrument class (ii): network-weighted characteristics.}
As in the previous case, fix an arbitrary $\boldsymbol{\theta}_1$ and $\bG$. Since $\mathbf{W}(\bG)$ is nonrandom conditional on $\bG$, we have that\footnote{Let $\mathbf{a}$ and $\mathbf{b}$ be random vectors and let $\bW$ be fixed given $\bG$. Decomposing both vectors into conditional means and deviations, $\mathbf{a} = \expectation[\mathbf{a} \mid \bG] + \dot{\mathbf{a}}$ and $\mathbf{b} = \expectation[\mathbf{b} \mid \bG] + \dot{\mathbf{b}}$,
we have 
\begin{equation*}
    \begin{split}
        \expectation[\mathbf{a}^\intercal \bW^\intercal \mathbf{b} \mid \bG] &=  \expectation[\mathbf{a} \mid \bG]^\intercal \bW^\intercal \expectation[\mathbf{b} \mid \bG] + \expectation[\dot{\mathbf{a}}^\intercal \bW^\intercal \dot{\mathbf{b}} \mid \bG] \\
        &=  \expectation[\mathbf{a} \mid \bG]^\intercal \bW^\intercal \expectation[\mathbf{b} \mid \bG] + \trace\!\bigl(\bW^\intercal\, \covariance(\mathbf{b}, \mathbf{a} \mid \bG)\bigr).
    \end{split}
\end{equation*}}
\begin{equation*}
    \begin{split}
    \expectation\!\left[\mathbf{z}_{s}^\intercal
    (\boldy_s - \mathbf{X}_s\boldsymbol{\theta}_1) \mid \bG \right] &= \expectation[\bnx_s^\intercal \bW^\intercal \mathbf{v}_{1,s}(\boldsymbol{\theta}_1, \bG) \mid \bG] \\
    &= \expectation[\bnx_s \mid \bG]^\intercal \bW^\intercal\, \expectation[\mathbf{v}_{1,s}(\boldsymbol{\theta}_1, \bG) \mid \bG] \\
    &\quad + \trace\bigl(\bW^\intercal\, \covariance(\mathbf{v}_{1,s}(\boldsymbol{\theta}_1, \bG), \bnx_s \mid \bG)\bigr).
    \end{split}
\end{equation*}
Using the conditional moment function $\mathbf{V}_{2,s}(\boldsymbol{\theta}_2, \bG)$, where $\boldsymbol{\theta}_2 := (\beta, \gamma, \delta, \sigma_u^2)$ shares $(\beta, \gamma,\delta)$ with $\boldsymbol{\theta}_1$ and contains the true $\sigma_u^2$, we obtain
\begin{equation*}
    \begin{split}
    \covariance(\mathbf{v}_{1,s}(\boldsymbol{\theta}_1, \bG), \bnx_s \mid \bG)
    &= \expectation[\mathbf{V}_{2,s}(\boldsymbol{\theta}_2, \bG) \mid \bG]
    - \sigma_u^2(\gamma\bI + \delta\bG).
    \end{split}
\end{equation*}
Substituting and applying the law of iterated expectations gives
\begin{equation*}
    \begin{split}
        \expectation[\mathbf{z}_s^\intercal \mathbf{v}_{1,s}(\boldsymbol{\theta}_1, \bG)]
    &= \expectation\bigl[\expectation[\bnx_s \mid \bG]^\intercal \bW^\intercal\,
      \expectation[\mathbf{v}_{1,s}(\boldsymbol{\theta}_1, \bG) \mid \bG]\bigr] \\
    &\quad + \expectation\bigl[\trace\bigl(\bW^\intercal\,
      \expectation[\mathbf{V}_{2,s}(\boldsymbol{\theta}_2, \bG) \mid \bG]\bigr)\bigr] \\
    &\quad - \sigma_u^2\bigl(\gamma\,\expectation[\trace(\bW)]
    + \delta\,\expectation[\trace(\bW\bG^\intercal)]\bigr).
    \end{split}
\end{equation*}
Under the trace conditions of the proposition, the final term vanishes for every $(\gamma, \delta, \sigma_u^2)$. The moment condition is therefore a linear combination of the conditional mean and covariance restrictions, with weights $\bW\,\expectation[\bnx_s \mid \bG]$ and $\bW$, respectively, that are fixed given the network. Validity follows by evaluating at
$\boldsymbol{\theta}_1 = \boldsymbol{\theta}_1^{*}$: under
Assumptions~\ref{ass:measurementerror1} and \ref{ass:measurementerror2},
we have that  $\expectation[\mathbf{v}_{1,s}(\boldsymbol{\theta}_1^{*}, \bG) \mid
\bG] = \mathbf{0}$ and $\expectation[\mathbf{V}_{2,s}(\boldsymbol{\theta}_2^{*}, \bG) \mid \bG] = \mathbf{0}$ for every $\bG$. Hence $\expectation\!\left[\mathbf{z}_{s}^\intercal
    (\boldy_s - \mathbf{X}_s\boldsymbol{\theta}^*_1) \right] = 0$.

\subsection{Proofs of Corollaries~\ref{cor:lagged} and \ref{cor:features}}

\paragraph{Corollary~\ref{cor:lagged}.}
Set $\bW(\bG_s) = \bA_s^{(\tau)}$ with $\tau \geq 2$. First, $\trace(\bA_s^{(\tau)}) = 0$, since no individual is at distance $\tau \geq 2$ from itself. Second, $\trace(\bA_s^{(\tau)}\bG_s^\intercal) = \sum_{i,j} A^{(\tau)}_{s,ij}\, g_{s,ij} = 0$, since $A^{(\tau)}_{s,ij} = 1$ implies that $i$ and $j$ are at distance $\tau \geq 2$ and hence unlinked, so $g_{s,ij} = 0$. Both traces vanish for every network structure and therefore in expectation. The result follows from Part 2 of Proposition~\ref{prop:instrumentsdiffmatrix}.

\paragraph{Corollary~\ref{cor:features}.}
For the instrument $\mathbf{z}_{s} = \mathbf{t}_s$, set $\mathbf{w}(\bG_s) = \mathbf{t}_s$. Validity follows directly from Part 1 of Proposition~\ref{prop:instrumentsdiffmatrix}.

For the instrument $\mathbf{z}_{s} = (\mathbf{t}_s - \expectation[t]\mathbf{1})\circ\bnx_s$, set $\bW(\bG_s) = \diag(\mathbf{t}_s - \expectation[t]\mathbf{1})$. Since $\bW$ is diagonal and $\bG_s$ has zero diagonal, $\trace(\bW\bG_s^\intercal) = 0$ for every network structure. Moreover,
\begin{equation*}
    \expectation[\trace(\bW)] = \expectation\Bigl[\textstyle\sum_i(t_{si} - \expectation[t])\Bigr] = \expectation\Bigl[\textstyle\sum_i t_{si}\Bigr] - \expectation[N_s]\,\expectation[t] = 0,
\end{equation*}
since $\expectation[t]$ denotes the mean of the statistic across all individuals in the population of networks. The result follows from Part 2 of Proposition~\ref{prop:instrumentsdiffmatrix}.

For the instrument $\mathbf{z}_{s} = \mathbf{d}_s\circ(\mathbf{t}_s - \expectation[t]\mathbf{1}) \circ\bG_s\bnx_s$, set $\bW(\bG_s) = \diag\bigl(\mathbf{d}_s\circ(\mathbf{t}_s - \expectation[t]\mathbf{1})\bigr)\bG_s$. First, $\trace(\bW) = 0$ for every network structure, since $\bG_s$ has zero diagonal. Second, using $(\bG_s\bG_s^\intercal)_{ii} = 1/d_{si}$,
\begin{equation*}
    \expectation[\trace(\bW\bG_s^\intercal)]
    = \expectation\Bigl[\textstyle\sum_i d_{si}(t_{si} - \expectation[t])\,
    (\bG_s\bG_s^\intercal)_{ii}\Bigr]
    = \expectation\Bigl[\textstyle\sum_i (t_{si} - \expectation[t])\Bigr] = 0,
\end{equation*}
where the last equality follows as above. The result again follows from Part 2 of Proposition~\ref{prop:instrumentsdiffmatrix}.

\section{Additional results}
\subsection{Identification under moment aggregation}\label{app:aggregation}
Using Equations~\eqref{eq:restmean2} and \eqref{eq:restcovariance2}, the conditional moment functions are affine in the parameters:
\begin{equation*}
    \expectation[\mathbf{v}_{1,s}(\boldsymbol{\theta}_1, \bG) \mid \bG] = \underbrace{\begin{bmatrix}
        \mathbf{1} & \expectation[\bnx_s \mid \bG] & \bG \expectation[\bnx_s \mid \bG] & \bG\expectation[\boldy_s \mid \bG] 
    \end{bmatrix}}_{=: \mathbf{X}_1(\bG)}(\boldsymbol{\theta}_1^* - \boldsymbol{\theta}_1),
\end{equation*}
and 
\begin{equation*}
    \expectation[\mathbf{v}_{2,s}(\boldsymbol{\theta}_2, \bG) \mid \bG] = \underbrace{\begin{bmatrix}
        \vect(\variance(\bnx_s \mid \bG))^\intercal \\
        \vect(\bG\variance(\bnx_s \mid \bG))^\intercal \\
        \vect(\bG \covariance(\boldy_s, \bnx_s \mid \bG))^\intercal \\
        \vect(-\bI)^\intercal \\
        \vect(-\bG)^\intercal
    \end{bmatrix}^\intercal}_{=: \mathbf{X}_2(\bG)} (\boldsymbol{\psi}(\boldsymbol{\theta}_2^*) - \boldsymbol{\psi}(\boldsymbol{\theta}_2)),
\end{equation*}
where $\boldsymbol{\psi}(\boldsymbol{\theta}_2) := (\gamma, \delta, \beta, \gamma \sigma_u^2, \delta \sigma_u^2)$ is a reparameterization of $\boldsymbol{\theta}_2$. Because these functions are linear in $\boldsymbol{\theta}_1$ and $\boldsymbol{\psi}$, respectively, so are their aggregates:
    \begin{equation*}
        \begin{split}
            \mathbf{h}_1(\boldsymbol{\theta}_1) &= \expectation\!\big[\mathbf{H}_1(\bG)^\intercal \expectation[\mathbf{v}_{1,s}(\boldsymbol{\theta}_1, \bG) \mid \bG]\big] = \mathbf{D}_1(\boldsymbol{\theta}_1^* - \boldsymbol{\theta}_1), \\
            \mathbf{h}_2(\boldsymbol{\theta}_2) &= \expectation\!\big[\mathbf{H}_2(\bG)^\intercal \expectation[\mathbf{v}_{2,s}(\boldsymbol{\theta}_2, \bG) \mid \bG]\big] = \mathbf{D}_2\left(\boldsymbol{\psi}(\boldsymbol{\theta}_2^*) - \boldsymbol{\psi}(\boldsymbol{\theta}_2)\right), 
        \end{split}
    \end{equation*}
where $\mathbf{D}_m := \expectation[\mathbf{H}_m(\bG)^\intercal \mathbf{X}_m(\bG)]$. Since $\mathbf{D}_1$ and $\mathbf{D}_2$ have dimensions $K_1 \times 4$ and $K_2 \times 5$, the rank conditions below require $K_1 \geq 4$ and $K_2 \geq 5$: aggregation must retain at least as many moments as parameters.

\begin{lemma} ~
    \begin{enumerate}
        \item Suppose that $\boldsymbol{\theta}_1^* \in \Theta_1$. Then $\boldsymbol{\theta}_1^*$ is the unique solution to $\mathbf{h}_1(\boldsymbol{\theta}_1) = \mathbf{0}$ if and only if $\rank(\mathbf{D}_1) = 4$. 
        \item Suppose that $\boldsymbol{\theta}_2^* \in \Theta_2$. Then $\boldsymbol{\theta}_2^*$ is the unique solution to $\mathbf{h}_2(\boldsymbol{\theta}_2) = \mathbf{0}$ if $\rank(\mathbf{D}_2) = 5$. 
    \end{enumerate}
\end{lemma}

\begin{proof}
    Since $\boldsymbol{\theta}_1^*$ always solves $\mathbf{h}_1(\boldsymbol{\theta}_1)=\mathbf{0}$, the set of solutions is a singleton if and only if $\ker(\mathbf{D}_1) = \{\mathbf{0}\}$, that is, if and only if $\rank(\mathbf{D}_1) = 4$.

    If $\rank(\mathbf{D}_2) = 5$, then $\ker(\mathbf{D}_2) = \{\mathbf{0}\}$ and $\mathbf{h}_2(\boldsymbol{\theta}_2)=\mathbf{0}$ forces $\boldsymbol{\psi}(\boldsymbol{\theta}_2^*) = \boldsymbol{\psi}(\boldsymbol{\theta}_2)$. Matching the first three components gives $(\beta^*, \gamma^*, \delta^*) = (\beta, \gamma, \delta)$. Because $\beta^*\gamma^* + \delta^* \neq 0$ implies $(\gamma^*, \delta^*) \neq (0,0)$, at least one of the remaining components, $\gamma^*{\sigma_u^2}^* = \gamma\sigma_u^2$ or $\delta^*{\sigma_u^2}^* = \delta\sigma_u^2$, can be solved for ${\sigma_u^2}^* = \sigma_u^2$. Hence $\boldsymbol{\theta}_2^* = \boldsymbol{\theta}_2$.
\end{proof}

%%%%%%%%%%%%%%%%%%%%%%%%%%%%%%%%%%%%%%%%
%%%%%%%%%%%%%%%%%%%%%%%%%%%%%%%%%%%%%%%%
%%%%%%%%%%%%%%%%%%%%%%%%%%%%%%%%%%%%%%%%
% ONLINE APPENDIX
%%%%%%%%%%%%%%%%%%%%%%%%%%%%%%%%%%%%%%%%
%%%%%%%%%%%%%%%%%%%%%%%%%%%%%%%%%%%%%%%%
%%%%%%%%%%%%%%%%%%%%%%%%%%%%%%%%%%%%%%%%
\clearpage

\setcounter{page}{1}
\renewcommand{\thepage}{O.A\arabic{page}}

\begin{center}
    {\huge Online Appendix \\
    \vspace{10pt}
    \Large \textit{Measurement Error and Peer Effects in Networks}
  }
\end{center}

\appendix
\renewcommand{\thesection}{O.\Alph{section}}

\section{Extensions}\label{app:extensions}

\subsection{Network fixed effects}\label{app:fixedeffects}
We consider an extended specification of the model that allows for network-specific fixed effects,
\begin{subequations}\label{eq:model_extended}
\begin{align}
\boldy_s &=  \alpha_s \mathbf{1} +  \gamma \bx_s + \delta \bG_s {\bx}_s + \beta \bG_s\boldy_s + \be_s, \label{eq:modelnetworkfixedeffect}\\
\bnx_s &= \bx_s + \bu_s,
\end{align}
\end{subequations}
with $\expectation(e_{si} \mid \alpha_s, \bx_s, \boldG_s, \bu_s) = 0$. The inclusion of the fixed effects $\alpha_s$ requires updated versions of Assumptions~\ref{ass:measurementerror1} and \ref{ass:measurementerror2}.
\begin{assumption}\label{ass:measurementerror1B}
    The measurement errors satisfy:
    \begin{equation*}
        \begin{split}
             \expectation[u_{si} \mid \alpha_s, \bx_s, \boldG_s, \be_s] &= 0.\\
        \end{split}
    \end{equation*}
\end{assumption}

\begin{assumption}\label{ass:measurementerror2B}
    The measurement errors satisfy:
    \[
\begin{alignedat}{2}
\mathbb{E}\!\left[u_{si}^{2}\mid \alpha_s,\bx_s,\boldG_s,\be_s\right] &= \sigma_u^{2}, &\qquad
\mathbb{E}\!\left[u_{si}u_{sj}\mid \alpha_s,\bx_s,\boldG_s,\be_s\right] &= 0, \forall i \neq j.
\end{alignedat}
\]
\end{assumption}

\paragraph{Identification.}
We now state analogues of Theorems~\ref{theorem:identification} and \ref{thm:identification2} for the model with network fixed effects. As in \citeA{bramoulle09}, the inclusion of network-specific fixed effects renders the identification conditions more demanding.

\begin{theorem}\label{theorem:identificationB}
    Suppose that Assumption~\ref{ass:measurementerror1B} holds and that $\beta \gamma + \delta \neq 0$. The parameters $\beta, \gamma, \delta$ are identified from conditional mean restrictions if the vectors $\bm$, $\bG_0 \bm$, $\bG_0^2\bm$, $\bG_0^3 \bm$ are linearly independent.
\end{theorem}

\begin{proof}
    Let $\boldsymbol{\theta}_3 := (\beta, \gamma,\delta)$ and let $\Theta_3 := \{ \boldsymbol{\theta}_3 : |\beta| < 1, \beta \gamma + \delta \neq 0\}$ denote the admissible parameter space. Suppose that $\boldsymbol{\theta}_3, \boldsymbol{\theta}'_3 \in \Theta_3$ generate the same locally differenced conditional mean restriction, that is,
    \begin{equation*}
    \begin{split}
        \expectation[(\bI - \bG_0)\boldy_0 \mid \bG_0] &= (\bI - \beta \bG_0)^{-1} (\bI - \bG_0) (\gamma \bI + \delta \bG_0) \bm, \\
            \expectation[(\bI - \bG_0)\boldy_0 \mid \bG_0] &= (\bI - \beta' \bG_0)^{-1} (\bI - \bG_0) (\gamma' \bI + \delta' \bG_0) \bm,
    \end{split}
\end{equation*}
where we made use of the push-through identity. Equivalently,
\begin{equation*}
    (\bI - \beta \bG_0)^{-1} (\bI - \bG_0) (\gamma \bI + \delta \bG_0) \bm = (\bI - \beta' \bG_0)^{-1} (\bI - \bG_0) (\gamma' \bI + \delta' \bG_0) \bm.
\end{equation*}
Left-multiplying both sides by $(\bI-\beta \bG_0)(\bI - \beta' \bG_0)$ yields
\begin{equation*}
    (\bI - \beta' \bG_0) (\bI - \bG_0) (\gamma \bI + \delta \bG_0) \bm = (\bI - \beta \bG_0) (\bI - \bG_0) (\gamma' \bI + \delta' \bG_0) \bm.
\end{equation*}
Expanding and rearranging gives $a_1 \bm + a_2 \bG_0 \bm  + a_3 \bG_0^2\bm + a_4\bG_0^3\bm = \mathbf{0}$, with coefficients
\[
\begin{alignedat}{2}
a_1&= \gamma-\gamma' ,&\qquad a_2&= -(1+\beta')\gamma+\delta+(1+\beta)\gamma'-\delta',\\
a_3&=\beta' \gamma - \delta(1+\beta') - \beta \gamma' + \delta'(1+\beta),&\qquad a_4&=\beta'\delta-\beta\delta'.
\end{alignedat}
\]
Collect these coefficients in the vector $A_3(\boldsymbol{\theta}_3, \boldsymbol{\theta}'_3) := (a_1, a_2, a_3, a_4)$ and define $\mathcal{A}_3 := \{\mathbf{c} \in \mathbb{R}^4 : c_1 \bm + c_2 \bG_0 \bm  + c_3 \bG_0^2\bm + c_4\bG_0^3\bm = \mathbf{0}\}$. If  $\boldsymbol{\theta}_3$ and $\boldsymbol{\theta}'_3$ generate the same differenced mean restriction then $A_3(\boldsymbol{\theta}_3, \boldsymbol{\theta}'_3) \in \mathcal{A}_3$.

Suppose that $\bm$, $\bG_0 \bm$, $\bG_0^2 \bm$, $\bG_0^3 \bm$ are linearly independent, so that $\mathcal{A}_3 = \{\mathbf{0}\}$ and hence $A_3(\boldsymbol{\theta}_3, \boldsymbol{\theta}'_3) = \mathbf{0}$. As in the proof of Proposition~1 in \citeA{bramoulle09}, $a_1 = a_2 = a_4 = 0$ together with $\beta \gamma + \delta \neq 0$ implies $(\beta, \gamma, \delta) = (\beta', \gamma', \delta')$. Hence, $\boldsymbol{\theta}_3$ is identified.
\end{proof}

\begin{theorem}\label{thm:identification2B}
    Suppose that Assumptions~\ref{ass:measurementerror1B} and \ref{ass:measurementerror2B} hold and that $\beta \gamma + \delta \neq 0$. The parameters $\beta, \gamma, \delta, \sigma_u^2$ are identified from conditional covariance restrictions if the matrices $\bI$, $\bG_0$, $\bG_0^2$, $\bG_0^3$, $\bC$, $\bG_0\bC$, $\bG_0^2 \bC$, $\bG_0^3 \bC$ are linearly independent.
\end{theorem}

\begin{proof}
    Suppose that $\boldsymbol{\theta}_2, \boldsymbol{\theta}'_2 \in \Theta_2$ generate the same locally differenced conditional covariance restriction, that is,
    \begin{equation*}
    \begin{split}
        \covariance((\bI - \bG_0)\boldy_0, \bnx_0 \mid \bG_0) &= (\bI - \beta \bG_0)^{-1} (\bI - \bG_0)(\gamma \bI + \delta \bG_0) (\bC - \sigma_u^2 \bI), \\
        \covariance((\bI - \bG_0)\boldy_0, \bnx_0 \mid \bG_0) &= (\bI - \beta' \bG_0)^{-1} (\bI - \bG_0) (\gamma' \bI + \delta' \bG_0) (\bC - {\sigma_u^2}' \bI),
    \end{split}
\end{equation*}
where we made use of the push-through identity. Equivalently, 
\begin{equation*}
    (\bI - \beta \bG_0)^{-1} (\bI - \bG_0)(\gamma \bI + \delta \bG_0) (\bC - \sigma_u^2 \bI) = (\bI - \beta' \bG_0)^{-1} (\bI - \bG_0) (\gamma' \bI + \delta' \bG_0) (\bC - {\sigma_u^2}' \bI).
\end{equation*}
Left-multiplying both sides by $(\bI-\beta \bG_0)(\bI - \beta' \bG_0)$ gives
\begin{equation*}
    (\bI - \beta' \bG_0) (\bI - \bG_0)(\gamma \bI + \delta \bG_0) (\bC - \sigma_u^2 \bI) = (\bI - \beta \bG_0) (\bI - \bG_0) (\gamma' \bI + \delta' \bG_0) (\bC - {\sigma_u^2}' \bI).
\end{equation*}
Expanding and rearranging gives $a_1 \bI + a_2 \bG_0 + a_3 \bG_0^2 + a_4 \bG_0^3 + a_5 \bC + a_6 \bG_0 \bC + a_7 \bG_0^2 \bC + a_8 \bG_0^3 \bC = \mathbf{0}$, with coefficients
\[
\begin{aligned}
a_1 &= -\sigma_u^2\gamma + {\sigma_u^2}'\gamma', \\
a_2 &= -\sigma_u^2(\delta - \gamma(1+\beta')) + {\sigma_u^2}'(\delta' - \gamma'(1+\beta)), \\
a_3 &= -\sigma_u^2(\gamma \beta' - (1+\beta')\delta) + {\sigma_u^2}'(\gamma' \beta - (1+\beta)\delta'), \\
a_4 &= -\sigma_u^2\beta'\delta + {\sigma_u^2}'\beta\delta', \\
a_5 &= \gamma - \gamma', \\
a_6 &= \delta - \gamma(1+\beta') - \delta' + \gamma'(1+\beta), \\
a_7 &= \gamma \beta' - (1+\beta')\delta - \gamma' \beta + (1+\beta)\delta', \\
a_8 &= \beta'\delta - \beta \delta'.
\end{aligned}
\]
Collect these coefficients in the vector $A_4(\boldsymbol{\theta}_2, \boldsymbol{\theta}'_2) := (a_1, a_2, a_3, a_4, a_5, a_6, a_7, a_8)$ and define $\mathcal{A}_4 := \{\mathbf{c} \in \mathbb{R}^8 : c_1 \bI + c_2 \bG_0 + c_3 \bG_0^2 + c_4 \bG_0^3 + c_5 \bC + c_6 \bG_0 \bC + c_7 \bG_0^2 \bC + c_8 \bG_0^3 \bC = \mathbf{0}\}$. If $\boldsymbol{\theta}_2$ and $\boldsymbol{\theta}'_2$ generate the same differenced covariance restriction then $A_4(\boldsymbol{\theta}_2, \boldsymbol{\theta}'_2) \in \mathcal{A}_4$.

Suppose that $\bI$, $\bG_0$, $\bG_0^2$, $\bG_0^3$, $\bC$, $\bG_0\bC$, $\bG_0^2 \bC$, $\bG_0^3 \bC$ are linearly independent, so that $\mathcal{A}_4 = \{\mathbf{0}\}$ and hence $A_4(\boldsymbol{\theta}_2, \boldsymbol{\theta}'_2) = \mathbf{0}$. As in the proof of Proposition~1 in \citeA{bramoulle09}, $a_5 = a_6 = a_8 = 0$ together with $\beta \gamma + \delta \neq 0$ implies $(\beta, \gamma, \delta) = (\beta', \gamma', \delta')$. Substituting into $a_1$ and $a_2$ gives
\begin{equation*}
    a_1 = -\gamma\,(\sigma_u^2 - {\sigma_u^2}') = 0,
    \qquad
    a_2 = -(\delta - \gamma(1 + \beta))\,(\sigma_u^2 - {\sigma_u^2}') = 0.
\end{equation*}
Suppose that $\sigma_u^2 \neq {\sigma_u^2}'$. Then $\gamma = 0$ and
$\delta - \gamma(1 + \beta) = 0$, hence $\delta = 0$ and
$\beta\gamma + \delta = 0$, contradicting the maintained assumption.
Therefore $\sigma_u^2 = {\sigma_u^2}'$. Hence, $\boldsymbol{\theta}_2$ is identified.
\end{proof}

\paragraph{Estimation.}
2SLS estimation carries over with adaptations. Let $\mathbf{Q}_s$ be a differencing matrix representing either network demeaning, $\bI_{N_s} - \frac{1}{N_s} \bJ_{N_s}$, or local differencing, $\bI_{N_s} - \bG_s$. The trace conditions of Proposition~\ref{prop:instrumentsdiffmatrix} become, for validity in the transformed model, $\expectation[\trace(\mathbf{Q}^\intercal_s\bW(\bG_s))] = \expectation[\trace(\mathbf{Q}^\intercal_s\bW(\bG_s)\bG_s^\intercal)] = 0$, with the instruments entering untransformed. Under local differencing, the argument of Corollary~\ref{cor:lagged} delivers valid instruments constructed from characteristics at distance three or larger. Under network demeaning, the network-lagged instruments are valid only as $N_s \to \infty$.\footnote{The moment condition carries a bias of order $\frac{\sigma_u^2}{N_s}$ because the network mean of the instrument contains the measurement errors of all individuals in the network.} A recentering approach in the spirit of Corollary~\ref{cor:features} is needed to restore validity at fixed $N_s$.

By contrast, GMM estimation is less straightforward, since every fixed effect $\alpha_s$ enters the nonlinear objective as an additional parameter, so that the parameter count grows with the number of networks. A simple alternative is to estimate the model using network-demeaned variables, which removes the fixed effects from the objective and can substantially facilitate estimation.

For concreteness, we use network demeaning, $\mathbf{Q}_s := \bI_{N_s} - \frac{1}{N_s} \bJ_{N_s}$; local differencing works similarly. Left-multiplying both sides of  Equation~\eqref{eq:modelnetworkfixedeffect} by $\mathbf{Q}_s$ yields the demeaned model
\begin{equation*}
    \begin{split}
        \mathbf{Q}_s\boldy_s &=  \mathbf{Q}_s\left[ (\gamma \bI + \delta \bG_s){\bx}_s + \beta \bG_s \boldy_s + \be_s \right],
    \end{split}
\end{equation*}
which no longer contains the network-specific fixed effect. From this demeaned model, under Assumption~\ref{ass:measurementerror1B}, we obtain the conditional mean restriction
\begin{equation*}
    \begin{split}
        \mathbf{Q}_0(\bI - \beta \bG_0)\expectation[\boldy_0 \mid \bG_0] &= \mathbf{Q}_0(\gamma \bI + \delta \bG_0) \expectation[\widetilde{\bx}_0 \mid \bG_0].
    \end{split}
\end{equation*}
For a fixed network $\bG$, the associated moment function for observation $s$ is
\begin{equation*}
    \mathbf{v}^{FE}_{1,s}(\boldsymbol{\theta}_1, \bG) := \mathbf{Q}_s \left[(\bI - \beta \bG)\boldy_s - (\gamma \bI + \delta \bG) \widetilde{\bx}_s \right],
\end{equation*}
where, in contrast to Equation~\eqref{eq:restmean2}, the parameter vector $\boldsymbol{\theta}_1 := (\gamma, \delta, \beta)$ no longer contains an intercept.
Likewise, under Assumptions~\ref{ass:measurementerror1B} and \ref{ass:measurementerror2B}, we obtain the conditional covariance restriction
\begin{equation*}
    \mathbf{Q}_0 (\bI - \beta \bG_0) \covariance(\boldy_0, \widetilde{\bx}_0 \mid \bG_0) \mathbf{Q}_0^\intercal = \mathbf{Q}_0 (\gamma \bI + \delta \bG_0) ( \variance(\widetilde{\bx}_0 \mid \bG_0) - \sigma_u^2 \bI)\mathbf{Q}_0^\intercal,
\end{equation*}
with associated moment function
\begin{equation*}
    \begin{split}
        \mathbf{V}^{FE}_{2,s}(\boldsymbol{\theta}_2, \bG) &:=  \mathbf{Q}_s \left[ (\bI - \beta \bG) \dot{\boldy}_s \dot{\widetilde{\bx}}^\intercal_s - (\gamma \bI + \delta \bG) ( \dot{\widetilde{\bx}}_s\dot{\widetilde{\bx}}^\intercal_s - \sigma_u^2 \bI) \right] \mathbf{Q}_s^\intercal.
    \end{split}
\end{equation*}
Note that $\mathbf{V}^{FE}_{2,s}$ involves two layers of demeaning. The dotted variables are deviations from conditional means, as in Equation~\eqref{eq:restcovariance2}, and form the conditional covariance; the sandwiching by $\mathbf{Q}_s$ demeans within networks and eliminates the fixed effect.

Together, these demeaned restrictions provide a potentially large system of equations in only four unknowns ($\beta$, $\gamma$, $\delta$, $\sigma_u^2$). GMM estimation then proceeds along the lines of Propositions~\ref{prop:GMMfirst} and \ref{prop:GMMsecond} in Section~\ref{sec:gmm}.

As a final remark, researchers who wish to include a larger set of perfectly measured covariates can follow \citeA{ericksonwhited2002}, who propose a computationally convenient two-step procedure. In the first step, one partials out fixed effects and the additional perfectly measured covariates via auxiliary OLS regressions. In the second step, one applies GMM to the resulting residualized model.\footnote{This approach generalizes the demeaning procedure outlined above. In this setting, $\mathbf{Q}_s$ can be interpreted as the annihilator matrix from the first stage.}

\subsection{Measurement error in outcomes}\label{app:outcomes}
We now extend the model to allow for mismeasured outcomes: the researcher observes $\widetilde{\boldy}_s = \boldy_s + \mathbf{v}_s$, where $\mathbf{v}_s$ collects the outcome measurement errors. Together with Equations~\eqref{eq:modeloutcome} and \eqref{eq:modelerror}, this yields
\begin{equation*}
    \widetilde{\boldy}_s = \alpha \mathbf{1} + \beta  \bG_s\widetilde{{\boldy}}_s + \gamma \bnx_s + \delta \bG_s {\bnx}_s + \underbrace{\mathbf{e}_s - \gamma \bu_s - \delta \bG_s{\bu}_s + \mathbf{v}_s - \beta\bG_s\mathbf{v}_s}_{=: \boldsymbol{\zeta}_s}.
\end{equation*}
Measurement error in outcomes is nontrivial in a linear-in-means setting, since outcomes enter both sides of the equation: the composite disturbance $\boldsymbol{\zeta}_s$ contains not only $\mathbf{v}_s$ but also the term $-\beta\bG_s\mathbf{v}_s$, through which each individual's error propagates to their peers. We show that, nevertheless, our identification and estimation results carry over once this additional source of error is disciplined, in the spirit of Assumptions~\ref{ass:measurementerror1} and \ref{ass:measurementerror2}.
\begin{assumption}\label{ass:measurementerror3B}
    The measurement errors satisfy:
    \[
    \begin{split}
            \mathbb{E}\!\left[u_{si}\mid \bx_s,\boldG_s,\be_s\right] &= 0, \\
            \mathbb{E}\!\left[v_{si}\mid \bx_s,\boldG_s,\be_s\right] &= 0.
    \end{split}
    \]
\end{assumption}
\noindent Assumption~\ref{ass:measurementerror3B} extends Assumption~\ref{ass:measurementerror1} with the analogous mean independence condition on the outcome measurement error.
\begin{assumption}\label{ass:measurementerror4B}
    The measurement errors satisfy:
\[
\begin{alignedat}{2}
\mathbb{E}\!\left[u_{si}^2 \mid \bx_s,\boldG_s,\be_s\right] &= \sigma_u^2, &\qquad
\mathbb{E}\!\left[u_{si} u_{sj}\mid \bx_s,\boldG_s,\be_s\right] &= 0, \forall i \neq j, \\
\mathbb{E}\!\left[u_{si} v_{sj} \mid \bx_s,\boldG_s,\be_s\right] &= 0, \forall i,j.
\end{alignedat}
\]
\end{assumption}
\noindent Assumption~\ref{ass:measurementerror4B} extends Assumption~\ref{ass:measurementerror2} by requiring the two sources of measurement error to be conditionally uncorrelated. Notably, we place no restrictions on the second moments of $v_{si}$: the outcome measurement error may be heteroscedastic and correlated across individuals within a network.

\paragraph{Identification.}
The identification conditions of Theorems~\ref{theorem:identification} and \ref{thm:identification2} carry over unchanged. The only difference is that the mean and covariance restrictions are now stated in terms of observed outcomes: $\expectation[\widetilde{\boldy}_0 \mid \bG_0]$ in place of Equation~\eqref{eq:restmean}, and $\covariance(\widetilde{\boldy}_0, \bnx_0 \mid \bG_0)$ in place of Equation~\eqref{eq:restcovariance}.
\begin{theorem}\label{theorem:identificationBB}
    Suppose that Assumption~\ref{ass:measurementerror3B} holds and that $\beta \gamma + \delta \neq 0$. The parameters $\alpha$, $\beta$, $\gamma$, $\delta$ are identified from conditional mean restrictions if and only if the vectors $\mathbf{1}$, $\bm$, $\bG_0 \bm$, $\bG_0^2 \bm$ are linearly independent.
\end{theorem}
\begin{proof}
    Assumption~\ref{ass:measurementerror3B} implies $\expectation[\mathbf{v}_0 \mid \bG_0] = \mathbf{0}$, so that
\begin{equation*}
    \expectation[\widetilde{\boldy}_0 \mid \bG_0]
= \expectation[\boldy_0 \mid \bG_0]
= (\bI - \beta \bG_0)^{-1} \bigl[\alpha \mathbf{1} + (\gamma \bI + \delta \bG_0)\bm\bigr].
\end{equation*}
This conditional mean restriction is identical in form to Equation~\eqref{eq:restmean}, and the argument in the proof of Theorem~\ref{theorem:identification} applies verbatim.
\end{proof}

\begin{theorem}\label{thm:identification2BB}
    Suppose that Assumptions~\ref{ass:measurementerror3B} and \ref{ass:measurementerror4B} hold and that $\beta \gamma + \delta \neq 0$. The parameters $\beta$, $\gamma$, $\delta$, $\sigma_u^2$ are identified from conditional covariance restrictions if the matrices $\bI$, $\bG_0$, $\bG_0^2$, $\bC$, $\bG_0\bC$, $\bG_0^2 \bC$ are linearly independent.
\end{theorem}
\begin{proof}
    Under Assumptions~\ref{ass:measurementerror3B} and \ref{ass:measurementerror4B},
\begin{equation*}
\covariance(\mathbf{v}_0, \bnx_0 \mid \bG_0) = \covariance(\mathbf{v}_0, \bx_0 \mid \bG_0) + \covariance(\mathbf{v}_0, \bu_0 \mid \bG_0) = \mathbf{0},
\end{equation*}
so that
\begin{equation*}
\covariance(\widetilde{\boldy}_0, \bnx_0 \mid \bG_0) = \covariance(\boldy_0, \bnx_0 \mid \bG_0) = (\bI - \beta \bG_0)^{-1} (\gamma \bI + \delta \bG_0) (\bC - \sigma_u^2 \bI).
\end{equation*}
This conditional covariance restriction is identical in form to Equation~\eqref{eq:restcovariance}, and the argument in the proof of Theorem~\ref{thm:identification2} applies verbatim.
\end{proof}

\paragraph{Estimation.}
Estimation is likewise unaffected. Under Assumption~\ref{ass:measurementerror3B}, the GMM estimator based on the mean restrictions coincides with that in Proposition~\ref{prop:GMMfirst}; under Assumptions~\ref{ass:measurementerror3B} and \ref{ass:measurementerror4B}, the GMM estimator based on the covariance restrictions coincides with that in Proposition~\ref{prop:GMMsecond}. Since the 2SLS moments are a linear combination of mean and covariance restrictions, Proposition~\ref{prop:instrumentsdiffmatrix} and Corollaries~\ref{cor:lagged} and \ref{cor:features} also remain valid.

\section{Additional results}\label{app2:addresults}

\subsection{Expansion bias without endogenous peer effect}\label{app2:biascontextual}
Consider the model without endogenous peer effect ($\beta = 0$), 
\begin{equation*}
\begin{split}
\boldy_s &=  \alpha \mathbf{1} + \gamma \bx_s + \delta \bG_s{\bx}_s + \be_s, \\
\bnx_s &= \bx_s + \bu_s,
\end{split}
\end{equation*}
whose parameters are estimated using the OLS estimator 
\begin{equation*}
    \begin{bmatrix}
    \gammaOLS & \deltaOLS
    \end{bmatrix}^\intercal := \left(\begin{bmatrix} \bnx_w & (\bG{\widetilde{\bx}})_w \end{bmatrix}^{\intercal}\begin{bmatrix} \bnx_w & (\bG{\widetilde{\bx}})_w  \end{bmatrix}\right)^{-1}\begin{bmatrix} \bnx_w & (\bG{\widetilde{\bx}})_w \end{bmatrix}^{\intercal} \boldy_w,
\end{equation*}
where, for any $N\times 1$ vector $\mathbf{a}$, the demeaned version is $\mathbf{a}_w := (\bI_N - \frac{1}{N}\bJ_N) \mathbf{a}$, with $\mathbf{I}_N$ the $N\times N$ identity and $\mathbf{J}_N$ the $N\times N$ matrix of ones.

Similar to Lemma~\ref{lemma:bias}, under Assumptions~\ref{ass:measurementerror1} and \ref{ass:measurementerror2} we have that
\begin{equation*}
    \text{plim }\begin{bmatrix}
        \gammaOLS \\ \deltaOLS
    \end{bmatrix} = (\mathbf{S} + \boldsymbol{\Sigma})^{-1} \mathbf{S} \begin{bmatrix}
        \gamma \\
        \delta
    \end{bmatrix},
\end{equation*}
where 
\begin{equation*}
    \begin{split}
        \mathbf{S} &:= \begin{bmatrix}
        \variance(x)  & \covariance(x, G{x}) \\
        \covariance(G{x}, x) & \variance(G{x})
        \end{bmatrix}, \qquad \boldsymbol{\Sigma} := \begin{bmatrix}
        \sigma_u^2 & 0 \\
        0 & h_0 \sigma_u^2 \\
    \end{bmatrix}.
    \end{split}
\end{equation*}
Therefore, it holds that
\begin{equation*}
    \begin{split}
         \text{plim }\begin{bmatrix}
        \gammaOLS \\ \deltaOLS
    \end{bmatrix} &= 
    \frac{1}{D^{OLS}}\begin{bmatrix}
        \variance(G{x}) + h_0 \sigma_u^2 & - \covariance(G{x}, x) \\
        -\covariance(x, G{x}) & \variance(x) + \sigma_u^2
    \end{bmatrix} \begin{bmatrix}
        \variance(x)  & \covariance(x, G{x}) \\
        \covariance(G{x}, x) & \variance(G{x})
        \end{bmatrix}
     \begin{bmatrix}
        \gamma \\
        \delta
    \end{bmatrix} \\
    &=  \frac{1}{D^{OLS}} \begin{bmatrix}
       \det(\mathbf{S}) + h_0 \sigma_u^2 \variance(x) & h_0 \sigma_u^2 \covariance(x, G{x}) \\ 
       \sigma_u^2 \covariance(x, G{x}) & \det(\mathbf{S}) + \sigma_u^2 \variance(G{x}) \\
    \end{bmatrix} \begin{bmatrix}
        \gamma \\
        \delta
    \end{bmatrix},
    \end{split}
\end{equation*}
where $D^{OLS}:= \det(\mathbf{S}) + h_0 \sigma_u^2 \variance(x) + \sigma_u^2 \variance(G{x}) + h_0 \sigma_u^4$. Matrix multiplication yields Equation~\eqref{eq:biasnoendogenous}. Observe that $\det(\mathbf{S}) = \variance(x)\variance(G{x}) - \covariance(x, G{x})^2 \geq 0$.

\subsection{Expansion bias without contextual peer effect}\label{app2:biasendogenous}
Consider the model without contextual peer effect ($\delta = 0$),
\begin{equation*}
\begin{split}
\boldy_s &=  \alpha \mathbf{1} + \gamma \bx_s + \beta \bG_s{\boldy}_s + \be_s, \\
\bnx_s &= \bx_s + \bu_s,
\end{split}
\end{equation*}
whose parameters are estimated using the 2SLS estimator
\begin{equation*}
    \begin{bmatrix}
    \gammaIV & \betaIV 
    \end{bmatrix}^\intercal := \left(\begin{bmatrix} \bnx_w & (\bG{\widetilde{\bx}})_w \end{bmatrix}^{\intercal}\begin{bmatrix} \bnx_w & (\bG{\boldy})_w \end{bmatrix}\right)^{-1}\begin{bmatrix} \bnx_w & (\bG{\widetilde{\bx}})_w \end{bmatrix}^{\intercal} \boldy_w,
\end{equation*}
where, for any $N\times 1$ vector $\mathbf{a}$, the demeaned version is $\mathbf{a}_w := (\bI_N - \frac{1}{N}\bJ_N) \mathbf{a}$, with $\mathbf{I}_N$ the $N\times N$ identity and $\mathbf{J}_N$ the $N\times N$ matrix of ones. 

Similar to Lemma~\ref{lemma:bias}, under Assumptions~\ref{ass:measurementerror1} and \ref{ass:measurementerror2} we have that
\begin{equation*}
    \text{plim }\begin{bmatrix}
        \gammaIV \\ \betaIV 
    \end{bmatrix}= (\mathbf{S} + \boldsymbol{\Sigma})^{-1} \mathbf{S} \begin{bmatrix}
        \gamma \\
        \beta
    \end{bmatrix},
\end{equation*}
where 
\begin{equation*}
    \begin{split}
        \mathbf{S} &:= \begin{bmatrix}
        \variance(x)  & \covariance(x, G{y}) \\
        \covariance(Gx, x) & \covariance(Gx, G{y})
        \end{bmatrix}, \qquad \boldsymbol{\Sigma} := \begin{bmatrix}
        \sigma_u^2 & 0 \\
        0 & 0 \\
    \end{bmatrix}.
    \end{split}
\end{equation*}
Therefore, it holds that
\begin{equation*}
    \begin{split}
         \text{plim }\begin{bmatrix}
        \gammaIV \\ \betaIV 
    \end{bmatrix} &= 
    \frac{1}{D^{IV}}\begin{bmatrix}
        \covariance(Gx, G{y}) & - \covariance(x,G{y}) \\
        -\covariance(Gx, x) & \variance(x) + \sigma_u^2
    \end{bmatrix} \begin{bmatrix}
 \variance(x)  & \covariance(x, G{y}) \\
        \covariance(Gx, x) & \covariance(Gx, G{y})
    \end{bmatrix}
     \begin{bmatrix}
        \gamma \\
        \beta
    \end{bmatrix} \\
    &= \frac{1}{D^{IV}}\begin{bmatrix}
        \variance(x)\covariance(Gx, G{y}) - \covariance(Gx, x) \covariance(x, G{y}) & 0 \\
        \sigma_u^2 \covariance(Gx, x) & D^{IV}
    \end{bmatrix} \begin{bmatrix}
        \gamma \\
        \beta
    \end{bmatrix},
    \end{split}
\end{equation*}
where $D^{IV} = (\variance(x) + \sigma_u^2) \covariance(Gx, G{y}) - \covariance(Gx, x) \covariance(x, G{y})$. Matrix multiplication yields Equation~\eqref{eq:biasnocontextual}.

\subsection{Regularity conditions for Propositions~\ref{prop:GMMfirst} and \ref{prop:GMMsecond}}\label{app:regularity}
Following \citeA[p.~172--174]{CameronTrivedi2005Microeconometrics}, we state the main regularity conditions ensuring consistency and asymptotic normality of the GMM estimators of Propositions~\ref{prop:GMMfirst} and \ref{prop:GMMsecond}.\footnote{Additional standard regularity conditions guarantee a uniform law of large numbers for the sample moments and their Jacobian.} Let $m \in \{1,2\}$ index the two sets of moments, corresponding respectively to the mean and covariance restrictions, and let $\boldsymbol{\theta}_m^*$ denote the corresponding true parameter value. For each $m$, assume:
\begin{itemize}
    \item[(i)] The model is correctly specified, i.e., $\expectation[\mathbf{v}_{m, s}(\boldsymbol{\theta}_m^*, \bG) \mid \bG] = \mathbf{0}$,
    \item[(ii)] The model is identified, i.e., the conditions in Theorem~\ref{theorem:identification} (for $m=1$) or Theorem~\ref{thm:identification2} (for $m=2$) hold,
    \item[(iii)] The Jacobian $\expectation\left[\frac{\partial \mathbf{v}_{m, s}(\boldsymbol{\theta}_m^*, \bG)}{\partial \boldsymbol{\theta}_m} \mid \bG\right]$ exists and is finite with full column rank,
    \item[(iv)] The moments are asymptotically normally distributed, i.e., 
    \begin{equation*}
        \frac{1}{{\sqrt{S}}} \sum_s \mathbf{v}_{m, s}(\boldsymbol{\theta}_m^*, \bG) \rightarrow_d N(\mathbf{0}, \boldsymbol{\Psi}_m(\boldsymbol{\theta}_m^*)),
    \end{equation*}    
    where $\boldsymbol{\Psi}_m(\boldsymbol{\theta}_m^*) := \expectation[\mathbf{v}_{m, s}(\boldsymbol{\theta}_m^*, \bG) \mathbf{v}_{m, s}(\boldsymbol{\theta}_m^*, \bG)^\intercal \mid \bG]$,
    \item[(v)] The parameter space $\boldsymbol{\Theta}_m$ is compact and $\boldsymbol{\theta}_m^* \in \mathrm{int}(\boldsymbol{\Theta}_m)$. 
\end{itemize}

\subsection{Relevance of instruments based on network features}\label{app:relevance}

\paragraph{The general result.}
We say that an instrument $z$ is (marginally) \emph{relevant} if $\covariance(z, \nx) \neq 0$ or $\covariance(z, G{\nx}) \neq 0$; joint relevance of the full instrument set requires $\plim \frac{1}{N}\mathbf{Z}^\intercal\mathbf{X}$ to have full column rank (see Section~\ref{sec:IV}). To ease notation, we write $t$, $d$, $x$, and $\nx$ for the entries of a representative individual, leaving the entrywise products of Corollary~\ref{cor:features} implicit.

\begin{proposition}\label{prop:relevance}
    Suppose that Assumptions~\ref{ass:measurementerror1} and \ref{ass:measurementerror2} hold. Let $t := t(\bG)$ be a network statistic. Then instruments (i) $t$, (ii) $(t - \mathbb{E}[t])\nx$, and (iii) $d(t - \mathbb{E}[t])G\nx$ are relevant if
    \begin{enumerate}
        \item[(i)] $\covariance(t, x) \neq 0 $ or $\covariance(t, G{x}) \neq 0$,
        \item[(ii)] $\covariance(t, x^2) - \mathbb{E}[x]\covariance(t, x) \neq 0$ or $\covariance(t, x G{ {x}}) - \mathbb{E}[G{{x}}]\covariance(t, x) \neq 0$,
        \item[(iii)] $\covariance(t, d x G{x}) - \mathbb{E}[x]\covariance(t, d G{x}) \neq 0$ or $\covariance(t, d (G{ {x}})^2) - \mathbb{E}[G{{x}}]\covariance(t, d G{x}) \neq 0$,
    \end{enumerate}
    respectively.
\end{proposition}

\begin{proof}
    For instrument (i), Assumption~\ref{ass:measurementerror1} implies $\covariance(t, \nx) = \covariance(t, x)$ and $\covariance(t, G{\nx}) = \covariance(t, G{x})$, so the stated condition delivers relevance directly.

     For instrument (ii), let $z = (t - \mathbb{E}[t])\nx$. By Assumptions~\ref{ass:measurementerror1} and \ref{ass:measurementerror2}, we have that
\begin{equation*}
    \begin{split}
        \covariance(z, \nx) &= \covariance(t \nx, \nx) - \mathbb{E}[t] \variance(\nx) \\
        &= \mathbb{E}[t \nx^2] - \mathbb{E}[t\nx] \mathbb{E}[\nx] - \mathbb{E}[t] \mathbb{E}[\nx^2] + \mathbb{E}[t]\mathbb{E}[\nx]^2 \\
        &= \covariance(t, \nx^2) - \mathbb{E}[\nx]\covariance(t, \nx) \\
        &= \covariance(t, x^2) - \mathbb{E}[x]\covariance(t, x),
    \end{split}
\end{equation*}
and
\begin{equation*}
    \begin{split}
        \covariance(z, G{\nx}) &= \covariance(t \nx, G{\nx}) - \mathbb{E}[t] \covariance(\nx, G{\nx}) \\
        &= \mathbb{E}[t \nx G{\nx}] - \mathbb{E}[t\nx] \mathbb{E}[G{\nx}] - \mathbb{E}[t] \mathbb{E}[\nx G{\nx}] + \mathbb{E}[t]\mathbb{E}[\nx] \mathbb{E}[G{\nx}] \\
        &= \covariance(t, \nx G{\nx}) - \mathbb{E}[G{\nx}]\covariance(t, \nx) \\
        &= \covariance(t, x G{x}) - \mathbb{E}[G{x}]\covariance(t, x).
    \end{split}
\end{equation*}
 
For instrument (iii), let $z = d(t - \mathbb{E}[t])G{\nx}$. By Assumptions~\ref{ass:measurementerror1} and \ref{ass:measurementerror2}, proceeding as above,
\begin{equation*}
    \begin{split}
        \covariance(z, \nx) &= \covariance(d t G{\nx}, \nx) - \mathbb{E}[t] \covariance(d G{\nx}, \nx) \\
        &= \covariance(t, d \nx G{\nx}) - \mathbb{E}[\nx]\covariance(t, d G{\nx}) \\
        &= \covariance(t, d x G{x}) - \mathbb{E}[x]\covariance(t, d G{x}),
    \end{split}
\end{equation*}
and
\begin{equation*}
    \begin{split}
        \covariance(z, G{\nx}) &= \covariance(d t G{\nx}, G{\nx}) - \mathbb{E}[t] \covariance(d G{\nx}, G{\nx}) \\
        &= \covariance(t, d (G{\nx})^2) - \mathbb{E}[G{\nx}]\covariance(t, d G{\nx}) \\
        &= \covariance(t, d (G{x})^2) - \mathbb{E}[G{x}]\covariance(t, d G{x}),
    \end{split}
\end{equation*}
where the last equality uses $\sum_j g_{ij}^2 = 1/d_i$, so that $\covariance(t, d(G{u})^2) = 0$.
\end{proof}

\paragraph{Iso-correlational environment.}
We now specialize to the undirected, iso-correlational environment of
Section~\ref{sec:IV}, in which $\expectation[x_{i_0} \mid \bG_0] = 0$, $\variance(x_{i_0} \mid \bG_0) = \sigma_x^2$, and $\covariance(x_{i_0}, x_{j_0} \mid \bG_0) = \rho_{d({i_0}, {j_0})}\sigma_x^2$. The zero-mean condition implies $\covariance(t, x) = \covariance(t, G{x}) = \expectation[Gx] = 0$. From the proof of Proposition~\ref{prop:relevance}, with $\expectation[Gx] = 0$ the recentering term drops out and $\covariance(d(t - \mathbb{E}[t])G{\nx}, G{\nx}) = \covariance(t, d (G{x})^2)$. Observe that
\begin{equation*}
    \begin{split}
    \covariance(t, d (G{x})^2) &= \left(\plim \frac{1}{N} \sum_i t_i d_i \sum_{j,k} g_{ij} g_{ik} x_j x_k\right) \\
    &\qquad - \left(\plim\frac{1}{N} \sum_i t_i \right)\left(\plim \frac{1}{N} \sum_i d_i \sum_{j,k} g_{ij} g_{ik} x_j x_k\right).
            \end{split}
\end{equation*}
Any two distinct neighbors of $i_0$ are at network distance one if they are linked and at distance two otherwise, so the iso-correlational structure implies
\begin{equation*}
    \begin{split}
        \plim \frac{1}{N} \sum_i t_i d_i \sum_{j,k} g_{ij} g_{ik} x_j x_k &= \frac{1}{N_0} \sum_{i_0} t_{i_0} d_{i_0} \left(\frac{1}{d_{i_0}} + \frac{(d_{i_0} -1) c_{i_0}}{d_{i_0}} \rho_1 + \frac{(d_{i_0} - 1)(1 - c_{i_0})}{d_{i_0}}\rho_2 \right) \sigma_x^2, \\
                \plim \frac{1}{N} \sum_i d_i \sum_{j,k} g_{ij} g_{ik} x_j x_k &= \frac{1}{N_0} \sum_{i_0}  d_{i_0} \left(\frac{1}{d_{i_0}} + \frac{(d_{i_0} -1) c_{i_0}}{d_{i_0}} \rho_1 + \frac{(d_{i_0} - 1)(1 - c_{i_0})}{d_{i_0}}\rho_2 \right) \sigma_x^2.
    \end{split}
\end{equation*}
Combining these expressions with $\plim\frac{1}{N} \sum_i t_i = \frac{1}{N_0} \sum_{i_0} t_{i_0}$ and collecting terms yields Equation~\eqref{eq:variancescovariances}.

\end{document}